\documentclass{entcs}
\usepackage{prentcsmacro}
\usepackage{graphicx}
\usepackage{amsmath,amssymb}
\usepackage{latexsym,upgreek}
\usepackage{xspace}
\usepackage{enumitem}
\usepackage{mathrsfs}
\usepackage{setspace}
\usepackage{color}

\long\def\COMMENT#1{}

\newcommand{\defAs}{\mathrel{\makebox{\:= \hspace{-.2cm} \raisebox{-0.5 ex}[0 ex][0 ex]{\tiny Def}}}}

\newcommand{\equivAs}{\mathrel{\makebox{\: $\equiv$ \hspace{-.2cm} \raisebox{-0.5 ex}[0 ex][0 ex]{\tiny Def}}}}

\newcommand{\comment}[1]{}

\def\eod {{\unskip\nobreak\hfil\penalty50
\hskip2em\hbox{}\nobreak\hfil $\Box$
\parfillskip=0pt \finalhyphendemerits=0 \par \medskip}}

\def\qed {{\unskip\nobreak\hfil\penalty50
\hskip2em\hbox{}\nobreak\hfil \rule{2mm}{2mm}
\parfillskip=0pt \finalhyphendemerits=0 \par \medskip}}

\newcommand{\model}{\ensuremath{\mbox{\boldmath $\mathcal{M}$}}\xspace}

\newcommand{\pow}{\mbox{\rm pow}}
\newcommand{\Rel}{\mbox{\rm Rel}}

\newcommand{\D}{\mathsf{D}}
\newcommand{\Sc}{\mathsf{S5}}

\newcommand{\TLQSR}{\ensuremath{\mbox{$\mathrm{3LQS}^{R}$}}\xspace}

\newcommand{\TLSZ}{\ensuremath{\mbox{$\mathrm{3LST}_{0}$}}\xspace}
\newcommand{\MLS}{\ensuremath{\mbox{$\mathrm{MLS}$}}\xspace}
\newcommand{\TLS}{\ensuremath{\mbox{$\mathrm{2LS}$}}\xspace}

\newcommand{\TLSSP}{\ensuremath{\mbox{$\mathrm{3LSSP}$}}\xspace}
\newcommand{\TLSSPU}{\ensuremath{\mbox{$\mathrm{3LSSPU}$}}\xspace}

\newcommand{\TLQST}{\ensuremath{\mbox{$\mathrm{3LQST_{0}}$}}\xspace}
\newcommand{\TLQSTR}{\ensuremath{\mbox{$\mathrm{3LQST_{0}}^{R}$}}\xspace}

\newcommand{\modelzs}{\model^{\vec{z},*}}
\newcommand{\modelsz}{\model^{*,\vec{z}}}
\newcommand{\modelZs}{\model^{\vec{Z},*}}
\newcommand{\modelsZ}{\model^{*,\vec{Z}}}
\newcommand{\modelz}{\model^{\vec{z}}}
\newcommand{\modelZ}{\model^{\vec{Z}}}
\newcommand{\modelStar}{\model^{*}}
\newcommand{\modelsZm}{\model^{*,\vec{Z}^{-}}}
\newcommand{\modelZms}{\model^{\vec{Z}^{-},*}}
\newcommand{\modelZp}{\model^{\vec{Z}'}}
\newcommand{\modelZm}{\model^{\vec{Z}^{-}}}
\newcommand{\modelZmz}{\model^{\vec{Z}^{-},\vec{z}}}

\newcommand{\assignz}{M^{\vec{z}}}

\newcommand{\assignStar}{M^{*}}

\newcommand{\assignZm}{M^{\vec{Z}^{-}}}
\newcommand{\assignZmz}{M^{\vec{Z}^{-},\vec{z}}}

\newcommand{\assignZms}{M^{\vec{Z}^{-},*}}

\newcommand{\dlss}{\mathcal{DL}\langle 4LQS^R\rangle(\D)}
\newcommand{\flqsr}{\ensuremath{4LQS^R}}

\renewcommand{\geq}{\geqslant}
\renewcommand{\leq}{\leqslant}

\newcommand{\xqedhere}[2]{%
  \rlap{\hbox to#1{\hfil\llap{\ensuremath{#2}}}}}

\def\lastname{Cantone, Nicolosi-Asmundo}

\begin{document}

\begin{frontmatter}
  \title{The decision problem for a three-sorted fragment of set theory with restricted quantification and finite enumerations} \author{Domenico Cantone\thanksref{ALL}\thanksref{cantone@dmi.unict.it}}
  \address{Dipartimento di Matematica e Informatica\\ Universit\`{a} di Catania\\
    Catania, Italy} \author{Marianna Nicolosi-Asmundo\thanksref{nicolosi@dmi.unict.it}}
  \address{Dipartimento di Matematica e Informatica\\ Universit\`{a} di Catania\\
    Catania, Italy} \thanks[ALL]{Thanks
    to everyone who should be thanked} \thanks[myemail]{Email:
    \href{cantone@dmi.unict.it} {\texttt{\normalshape
        cantone@dmi.unict.it}}} \thanks[coemail]{Email:
    \href{nicolosi@dmi.unict.it} {\texttt{\normalshape
        nicolosi@dmi.unict.it}}}
\begin{abstract} 
 We solve the satisfiability problem for a three-sorted fragment of set theory (denoted $\TLQSTR$), which admits a restricted form of quantification over individual and set variables and the finite enumeration operator $\{\text{-}, \text{-}, \ldots, \text{-}\}$ over individual variables, by showing that it enjoys a small model property, i.e., any satisfiable formula $\psi$ of $\TLQSTR$ has a finite model whose size depends solely on the length of $\psi$ itself.
Several set-theoretic constructs are expressible by $\TLQSTR$-formulae, such as  some variants of the power set operator and the unordered Cartesian product. In particular, concerning the unordered Cartesian product, we show that when finite enumerations are used to represent the construct, the resulting formula is exponentially shorter than the one that can be constructed without resorting to such terms. 
\end{abstract}
\begin{keyword}
  Please list keywords from your paper here, separated by commas.
\end{keyword}
\end{frontmatter}

\section{Introduction}
\emph{Computable set theory} is a research field 
%active since the late
%seventies.  Its initial goal was the design of effective
%decision procedures to be implemented in theorem provers/verifiers,
%for larger and larger 
studying the decidability of the satisfiability problem for collections of set-theoretic formulae (also
called \emph{syllogistics}).  
%During the years, however, due to the
%discovery of several decidability results of a purely theoretical
%nature, the main emphasis shifted to the foundational goal of
%narrowing the boundary between the decidable and the undecidable in
%set theory.

The main results in computable set theory up to 2001 have been
collected in \cite{CFO89,COO01}.  We also mention that the most
efficient decision procedures have been implemented in the proof
verifier \textsf{{\AE}tnaNova} \cite{SOC07} and form its inferential core.
%within one of
%versions of the system \textsf{STeP} \cite{STeP}.

%The basic set-theoretic fragment is the so-called Multi-Level
%Syllogistic (\MLS, for short) which involves in addition to
%variables, which are assumed to range over the von Neumann universe of
%sets, and to propositional connectives, also the basic set-theoretic
%operators such as $\cup$, $\cap$, $\setminus$, and the relators $=$,
%$\in$, $\subseteq$.  \MLS was proved decidable in \cite{FOS80} and
%extended over the years in several ways by the introduction of various
%operators, predicates, and restricted forms of quantification.

Most of the decidability results established in computable set theory regard
one-sorted multi-level syllogistics, namely collections of formulae
involving variables of one type only, ranging over the von Neumann
universe of sets.  On the other hand, few decidability results have
been proved for multi-sorted stratified syllogistics, admitting variables
of several types.  This, despite of the fact that in many
fields of computer science and mathematics often one deals with
multi-sorted languages.

An efficient decision procedure for the satisfiability of the
Two-Level Syllogistic language (\TLS), a fragment admitting
variables of two sorts for individuals and sets of individuals, basic set-theoretic
operators such as $\cup$, $\cap$, $\setminus$, the relators $=$,
$\in$, $\subseteq$, and propositional connectives, has
been presented in \cite{FerOm1978}.  Subsequently, in \cite{CanCut90},
the extension of \TLS with the singleton operator and the Cartesian
product operator has been proved decidable.  
%The latter decidability
%result has been obtained by embedding the elementary theory of
%relations in the class of purely universal formulae.  
Tarski's and
Presburger's arithmetics extended with sets have been studied in
\cite{CCS90}.  The three-sorted language \TLSSPU (Three-Level
Syllogistic with Singleton, Powerset, and general Union), allowing three types of variables,
and the singleton, powerset, and general union operators, in
addition to the operators and predicates already contained in \TLS,
 has been
proved decidable in \cite{CanCut93}.  
%More specifically, \TLSSPU has
%three types of variables, ranging over individuals, sets of
%individuals, and collections of sets of individuals, respectively, and
%involves the singleton, powerset, and general union operators, in
%addition to the operators and predicates already contained in \TLS.
More recently, in \cite{CanNic08}, the three-level quantified syllogistic \TLQSR, involving variables of three sorts  has been shown to have a
decidable satisfiability problem. Later, in \cite{CanNic13}, the satisfiability problem 
for \flqsr, a four-level quantified syllogistic admitting variables of four sorts has been proved to be decidable. The latter result has been exploited in \cite{CanLonNicSan15} to prove that $\dlss$, an expressive description logic, has the consistency problem for its knowledge bases decidable.  

In this paper we present a decidability result for the satisfiability
problem of the set-theoretic language $\TLQSTR$ (Three-Level Quantified
Syllogistic with Finite Enumerations and Restricted quantifiers), which is a three-sorted
quantified syllogistic
%(see
%\cite{FerOm1978,CanCut90,CCS90,CanCut93} for other examples of
%multi-sorted stratified syllogistics that have been proved
%decidable),
%whose language allows variables of different kinds.  In particular
%analogously to \TLSSPU,
involving \emph{individual variables}, \emph{set variables}, and
\emph{collection variables}, ranging respectively over the elements of
a given nonempty universe $D$, over the subsets of $D$, and 
over the collections of subsets of $D$. The language of $\TLQSTR$ admits  the predicate symbols $=$ and $\in$
and a restricted form of quantification
over individual and set variables. $\TLQSTR$ extends the fragment \TLQSR presented in \cite{CanNic08} since it admits
the finite enumeration operator $\{\text{-}, \text{-}, \ldots, \text{-}\}$ over individual variables.  In spite of its simplicity, $\TLQSTR$ allows
one to express several constructs of set theory.  Among them, the most
comprehensive one is the set former, which in turn allows one to
express other set-theoretic operators like several variants of the powerset and the unordered Cartesian product. We will present two different  $\TLQSTR$ representations of the latter construct: the first, more straightforward one involves finite enumerations and has linear length in the size of the unordered Cartesian product, the second one does not involve finite enumerations, is exponentially longer than the 
 first representation, and is  expressible also in  \TLQSR. 

We will prove that $\TLQSTR$ enjoys a small model property by showing
how to extract, out of a given model satisfying a $\TLQSTR$-formula
$\psi$, another model of $\psi$ but of bounded finite cardinality.
%%The construction of the finite model is inspired to the
%%algorithms described in %\cite{FerOm1978}, 
%%\cite{CanCut90,CanCut93}.
%%In
%%particular it extends the approach presented in \cite{CanCut90} to
%%handle quantifiers over set variables and it is simpler and more
%%efficient than the procedure introduced in \cite{CanCut93}.
%Then, we introduce the subtheories $(\TLQSTR)^h$ of $\TLQSTR$,
%consisting of $\TLQSTR$-formulae whose quantifier prefixes have length
%bounded by $h \geq 2$ and which satisfy certain additional syntactic
%constraints.  Each subtheory $(\TLQSTR)^h$ is shown to have an
%\textsf{NP}-complete satisfiability problem. Moreover, constructs such as variants of the powerset operator and a restricted form of the unordered cartesian product can be formalized in $(\TLQSTR)^h$, for some fixed $h \geq 2$,  and the normal modal
%logic $\Sc$ can be expressed by a $(\TLQSTR)^3$-formula. 
%%\marginpar{Aggiungere la parte su $\TLQSTR$.}

%\smallskip

The paper is organized as follows.  In Section \ref{genericlanguage},
we describe the syntax and semantics of a more general 
%three-level
%quantified 
language, denoted $\TLQST$, which contains $\TLQSTR$ as a
proper fragment.
%, and we also illustrate the expressiveness of the
%language $\TLQSR$ of our interest.
% 
Subsequently, in Section \ref{machinery} the machinery needed to prove
our main decidability result is provided.  
%In particular, a general
%definition of relativized interpretation is introduced, together
%with some useful technical results.  
In Section \ref{satisfiability},
the small model property for $\TLQSTR$ is established, thus solving the
satisfiability problem for $\TLQSTR$.  Then, in Section
\ref{expressiveness}, 
%\textcolor{red}{after illustrating the
%expressivity of $\TLQSR$,} 
we show how $\TLQSTR$ can be used to express several set theoretical operators.  
%that $(\TLQSR)^3$ can express the modal logic $\Sc$. 
%\marginpar{Aggiungere la parte su $\TLQSTR$.}
% In Section \ref{expressiveness} some hints of the
% expressive power of $3LQS$ are given.
Finally, in Section
\ref{conclusions}, we draw our conclusions.

% \section{The language $\TLQSR$}\label{language}
% We present the language $\TLQSR$ of our interest.  
%Finally, in Section \ref{expressiveness}, we illustrate with
%several examples the expressiveness of $\TLQSR$.

%We introduce the syntax and the semantics of the language $3LQS$.
%$3LQS$ is a lean language: it does not contain any operator symbol
%but only $=$ and $\in$ as predicate symbols. However, admitting
%quantification over individual variables and set variables, it can
%express the set former construct and, consequently, operators like
%the powerset
%
%
\section{The language $\TLQST$ and its subfragment $\TLQSTR$}\label{genericlanguage}

We begin by defining the syntax and the semantics of the more general
three-level quantified language $\TLQST$.  Then, in Section
\ref{restrictionquant}, we show how to characterize $\TLQSTR$-formulae
by suitable restrictions on the usage of quantifiers in formulae of
$\TLQST$.

%\subsection{Syntax of $\TLQST$} 
The three-level quantified language
$\TLQST$ involves
%\footnote{In the paper, variables often will come with
%numerical subscripts.  Other types of subscripts will be used in
%Section \ref{expressiveness} for variables denoting sets or
%collections of sets of particular relevance (e.g., $X_{U}$,
%$A_{< h}$).}
\begin{itemize}
\item[(i)] a collection $\mathcal{V}_{0}$ of \emph{individual} or 
\emph{sort $0$ variables},
denoted by $x,y,z,\ldots$;

\item[(ii)] a collection $\mathcal{V}_{1}$ of \emph{set} or 
\emph{sort $1$ variables},
denoted by $X,Y,Z,\ldots$;

\item[(iii)] a collection $\mathcal{V}_{2}$ of \emph{collection} or
\emph{sort $2$ variables}, denoted by 
$A,B,C,\ldots$.

\end{itemize}
%
%
% \begin{remark}
% Individual, set, and collection variables will be also referred to as
% variables of level 0, 1, and 2, respectively.
% \end{remark}
In addition to variables $\TLQST$ involves also \emph{finite enumerations} of type $\{x_1,\ldots,x_k\}$, with $x_1,\ldots,x_k \in \mathcal{V}_{0}$, $k > 0$. $\TLQST$-\emph{quantifier-free atomic formulae} are classified as:
\begin{itemize}
\item \emph{level $0$}: $x = y$, $x \in X$, $\{x_1,\ldots,x_k \} = X$, $\{x_1,\ldots,x_k \} \in A$, where
$x,y,x_1,\ldots,x_k \in \mathcal{V}_{0}$, $k>0$, $X \in \mathcal{V}_{1}$, and $A \in \mathcal{V}_{2}$;
% \begin{itemize}
%     \item  $x = y$, for $x,y \in \mathcal{V}_{0}$;
%     \item  $x \in X$, for $x \in \mathcal{V}_{0}, X \in \mathcal{V}_{1}$;
% \end{itemize}

\item \emph{level $1$}:  $X = Y$, $X \in A$, where $X,Y \in \mathcal{V}_{1}$ and $A \in \mathcal{V}_{2}$.
% \begin{itemize}
%     \item $X = Y$, for $X,Y \in \mathcal{V}_{1}$;
%     \item $X \in A$, for $X \in \mathcal{V}_{1}, A \in \mathcal{V}_{2}$;
%     \item $(\forall z_1) \ldots (\forall z_n) \varphi_0$, with
%     $\varphi_0$ a propositional combination of level $0$ atoms and $z_1,\ldots,z_n$ variables of sort $0$;
% \end{itemize}
\end{itemize}
$\TLQST$ \emph{purely universal formulae} are classified as:

\begin{itemize}
\item \emph{level $0$}: $(\forall z_1) \ldots (\forall z_n) \varphi_0$, 
with $\varphi_0$ a propositional combination of level $0$
quantifier-free atoms and $z_1,\ldots,z_n$ variables of sort $0$, with
$n \geq 1$;\footnote{The logical connectives admitted in propositional
combinations are the usual ones: negation $\neg$, conjunction
$\wedge$, disjunction $\vee$, implication $\rightarrow$, and 
biimplication $\leftrightarrow$.}

\item \emph{level $1$}: $(\forall Z_1) \ldots (\forall Z_m) \varphi_1$, 
where $\varphi_1$ is a propositional combination of quantifier-free
atomic formulae of any level and of purely universal formulae of level
$0$, and $Z_1,\ldots,Z_m$ are variables of sort $1$, with $m \geq 1$.
%    such that the formula
%    \[
%    \neg \varphi_0 \rightarrow \bigwedge_{j=1}^m \bigwedge_{i=1}^n z_i \in Z_j
%    \]
%    is valid for every level $1$ atom of the form
%    $(\forall z_1) \ldots (\forall z_n) \varphi_0$ present in
%    $\varphi_1$:
%    in this case we say that the atom
%    $(\forall z_1) \ldots (\forall z_n) \varphi_0$ is {\em linked} to
%    the variables $Z_1, \ldots , Z_m$; thus quantifiers
%    over individual variables nested within quantifiers over set
%    variables are restricted in some way);\footnote{The notion of
%    linked atoms of type $(\forall z_1) \ldots (\forall z_n)
%    \varphi_0$ will become clearer after the definition of the
%    intended semantics for $3LQS$ below.}
\end{itemize}
Finally, the \emph{formulae of $\TLQST$} are all the propositional
combinations of quantifier-free atomic formulae and of purely
universal formulae of levels $0$ and $1$.

% \noindent Some examples of $3LQS$ formulae can be found in Section
% \ref{expressiveness}.
%Notice that, quantification in this language is allowed only in a
%very strict manner. Quantifiers over individual variables are never
%nested since they can only be in front of quantifier free formulae.
%
%Single nestings are allowed only once, and between quantifiers over
%sets and over individuals.

%\subsection{Semantics of $\TLQST$} 
A \emph{$\TLQST$-interpretation}
is a pair $\model=(D,M)$, where $D$ is any nonempty collection of objects, called the
    \emph{domain} or \emph{universe} of $\model$, and $M$ is an assignment over the variables of $\TLQST$ such that
    \begin{itemize}
	\item $Mx \in D$, for each individual variable $x \in
	\mathcal{V}_{0}$;
    \item $MX \subseteq D$, for each set variable $X \in
	\mathcal{V}_{1}$;
     \item  $MA \subseteq \pow(D)$, for all collection variables
        $A \in \mathcal{V}_{2}$.\\
        (we recall that $\pow(s)$ denotes the powerset of $s$)
      %  \footnote{We recall that, for any set
   % $s$, $\pow(s)$ denotes the \emph{powerset} of $s$, i.e., the
 %   collection of all subsets of $s$.}
    \end{itemize}

%\begin{sloppypar}
Next, let

\smallskip
{- $\model=(D,M)$ be a $\TLQST$-interpretation,}

{- $x_1,....x_n \in \mathcal{V}_0$,~~ $X^1 _1, ... X^1 _m \in \mathcal{V}_1$,}

{- $u_1, ... u_n \in D$,~~ $U^1 _1, ... U^1 _m \in \pow(D)$.}

\smallskip
\noindent
%Given a $\TLQST$-interpretation $\model=(D,M)$, individual variables
%$z_{1},\ldots,z_{n} \in \mathcal{V}_{0}$, set variables
%$Z_{1},\ldots,Z_{m} \in \mathcal{V}_{1}$, individuals
%$u_{1},\ldots,u_{n} \in D$, and sets $U_{1},\ldots,U_{m} \subseteq D$,
By $\model[z_{1}/u_{1},\ldots,z_{n}/u_{n},
         Z_{1}/U_{1},\ldots,Z_{m}/U_{m}]$
we denote the $\TLQST$-interpretation $\model' = (D,M')$ such that $M'x_i =u_i$, for $i=1,...,n$, $M'X^1_j = U^1 _j$, for $j=1,...,m$, and which otherwise coincides with $M$ on all remaining variables.
%
%whose
%assignment $M'$ coincides with $M$ on all variables but
%$z_{1},\ldots,z_{n}$ and $Z_{1},\ldots,Z_{m}$, and 
%\[
%\begin{array}{llll}
%    M'z_{i} & = & u_{i}\,, & \mbox{{\rm for} $i = 1,\ldots,n$}  \\
%    M'Z_{j} & = & U_{j}\,, & \mbox{{\rm for} $j = 1,\ldots,m$}.  \\
%   % M'A_{k} & = & \mathcal{A}_{k}\,, & \mbox{{\rm for} $k = 1,\ldots,n$}
%\end{array}
%\]
%\end{sloppypar}

Throughout the paper we will use the abbreviations: $\modelz \defAs \model[z_{1}/u_{1},\ldots,z_{n}/u_{n}]$, $\modelZ \defAs \model[Z_{1}/U_{1},\ldots,Z_{m}/U_{m}]\,,$ 
%\[
%\begin{array}{rcl}
%\modelz & \defAs & \model[z_{1}/u_{1},\ldots,z_{n}/u_{n}]  \\
%\modelZ & \defAs & \model[Z_{1}/U_{1},\ldots,Z_{m}/U_{m}]\,,
%\end{array}
%\]
where the variables $z_{i}$ and $Z_{j}$, the individuals $u_{i}$, and the
subsets $U_{j}$ are understood from the context.

Let $\psi$ be a $\TLQST$-formula and let $\model = (D, M)$ be a
$\TLQST$-interpre\-ta\-tion. The notion of \emph{satisfiability}  for $\psi$ with respect to $\model$ (denoted
by $\model \models \psi$) is defined recursively over the structure of $\varphi$. The evaluation of quantifier-free atomic formulae is carried out as usual according to the standard meaning of the predicates `$\in$'
and `$=$'. Purely universal formulae are interpreted as follows:
%
%
%
%We define inductively the notion of {\em
%satisfiability} for $\psi$ with respect to $\model$ (denoted
%by $\model \models \psi$) as follows
%\begin{enumerate}[label=\arabic*.]
%\item $\model \models x=y$ ~~iff~~ $Mx = My$;
%\item $\model \models x \in X$ ~~iff~~ $Mx \in MX$;
%\item $\model \models \{x_1,\ldots,x_k\} = X$ ~~iff~~ $\{Mx_1,\ldots, Mx_k\} = MX$;
%\item $\model \models \{x_1,\ldots,x_k\} \in A$ ~~iff~~ $\{Mx_1,\ldots, Mx_k\} \in MA$;
%\item $\model \models X=Y$ ~~iff~~ $MX = MY$;
%\item $\model \models X \in A$ ~~iff~~ $MX \in MA$;
\begin{itemize}
\item $\model \models (\forall z_1) \ldots (\forall z_n) \varphi_0$
~~iff~~ $\model[z_1/u_1,\ldots , z_n/u_n] \models \varphi_0$, \\
for all $u_1,\ldots ,u_n \in D$;
%(propositional connectives are
%interpreted in the standard way);
\item  $\model \models (\forall Z_1) \ldots (\forall Z_m)
\varphi_1$ ~~iff~~ $\model[Z_1/U_1,\ldots , Z_m/U_m] \models
\varphi_1$, \\
for all $U_1,\ldots , U_n \subseteq D$.
\end{itemize}
Finally, compound formulae are evaluated according to the standard rules of propositional logic.
%Propositional connectives are interpreted in the standard way,
%namely
%% 
%\begin{enumerate}[label=\arabic*.,resume]
%\item $\model \models \varphi_{1} \wedge \varphi_{2}$ ~~iff~~
%$\model \models \varphi_{1}$ and $\model \models \varphi_{2}$;
%
%\item $\model \models \varphi_{1} \vee \varphi_{2}$ ~~iff~~
%$\model \models \varphi_{1}$ or $\model \models \varphi_{2}$;
%
%\item $\model \models \neg \varphi$ ~~iff~~
%$\model \not \models \varphi$. \eod
%\end{enumerate}
%\end{definition}
% Propositional connectives are treated in the standard way and
% therefore are not explicitly considered in the definition above.
Let $\psi$ be a $\TLQST$-formula.  If $\model \models \psi$ (i.e.,
$\model$ \emph{satisfies} $\psi$), then $\model$ is said to be a
\emph{$\TLQST$-model for $\psi$}.  A $\TLQST$-formula is said to be
\emph{satisfiable} if it has a $\TLQST$-model.  A $\TLQST$-formula is
\emph{valid} if it is satisfied by all $\TLQST$-interpretations.

\subsection{Characterizing the restricted fragment $\TLQSTR$}\label{restrictionquant}

$\TLQSTR$ is the collection of the $\TLQST$-formulae $\psi$ such that,
for \emph{every} purely universal formula $(\forall Z_1)\ldots(\forall
Z_m)\varphi_1$ of level 1 occurring in $\psi$ and \emph{every} purely
universal formula $(\forall z_1) \ldots (\forall z_n) \varphi_0$ of
level 0 occurring in $\varphi_1$, the condition
%\footnote{Compared to the restriction on quantifiers given in
%the version of the paper presented at CILC 2007, namely $\neg
%\varphi_0 \rightarrow \bigwedge_{j=1}^m \bigwedge_{i=1}^n z_i \in
%Z_j$, the condition introduced in this version of the paper is
%slightly weaker.}
\begin{equation}
    \label{condition}
{\small \neg \varphi_0 \rightarrow \bigwedge_{i=1}^n \bigwedge_{j=1}^m  z_i \in Z_j}
\end{equation}
is a valid $\TLQST$-formula (in this case we say that the purely
universal formula $(\forall z_1) \ldots (\forall z_n) \varphi_0$ is
\emph{linked to the variables $Z_1, \ldots , Z_m$}).

Condition~(\ref{condition}) guarantees that, if a given interpretation
assigns to $z_1,\ldots,z_n$ elements of the domain that make
$\varphi_0$ false, then all such values must be contained as elements
in the intersection of the sets assigned to $Z_1,\ldots,Z_m$.  This fact is used 
in the proof of Lemma \ref{quantifiedformTwo} 
to make sure that satisfiability is preserved in the finite model.  As
the examples in Section \ref{expressiveness} will illustrate,
condition~(\ref{condition}) is not particularly restrictive.

The following question arises: how one can establish whether a given
$\TLQST$-formula is a $\TLQSTR$-formula?  Observe that neither
quantification nor collection variables are involved in
condition~(\ref{condition}).  Indeed, it turns out that
(\ref{condition}) is a \TLS-formula and therefore one could use the
decision procedures in \cite{FerOm1978} to test its validity, as
$\TLQST$ is a conservative extension of \TLS. We mention also that in
most cases of interest, as will be shown in detail in Section
\ref{expressiveness}, condition~(\ref{condition}) is just an instance
of the simple propositional tautology $\neg (\mathbf{p} \rightarrow
\mathbf{q}) \rightarrow \mathbf{p}$, and therefore its validity can
follow just by inspection.

\section{Relativized interpretations}\label{machinery}
Small models of satisfiable $\TLQSTR$-formulae will be expressed in
terms of \emph{relativized interpretations} with respect to a suitable
domain.

% $\psi$
% 
% We introduce the notion of relativized interpretations: these will be 
% 
% , to be used
% together with the decision procedure of Section \ref{decisionproc} to
% construct, out of any model $\model = (D,M)$ for a $\TLQSR$-formula
% $\psi$, another model $\modelStar = (D^*,\assignStar )$ for $\psi$ but of finite
% bounded size.

%The introduction of a notion of relativized model, constructed as a
%substructure of a given model is a common practice with several
%logics and languages. Here we give a general notion of relativized
%model for the $3LQS$-language analyzing its main features. Such
%notion will be then applied to prove that the $3LQS$-language has
%the finite model property.
%
%All the considerations coming out from such analysis will then be
%used to define the finite model.
\begin{definition}[Relativized interpretations]\label{relintrp}
Let $\model=(D,M)$ be a $\TLQST$-interpretation and let $D^{*}
\subseteq D$, $d^{*} \in D^{*}$, $\mathcal{V}'_{0} \subseteq \mathcal{V}_{0}$, $\mathcal{V}'_{1} \subseteq \mathcal{V}_{1}$, and $l >0$.  The \emph{relativized interpretation} $\Rel(\model, D^{*},
d^{*}, \mathcal{V}'_{0}, \mathcal{V}'_{1}, l)$ of $\model$ with respect to $D^{*}$, $d^{*}$, $\mathcal{V}'_{0}$,
$\mathcal{V}'_{1}$, and $l$ is the interpretation $\modelStar = (D^{*},M^{*})$
such that
\begin{eqnarray*}
    M^{*}x & = & \begin{cases}
        Mx\,, & \mbox{if $Mx \in D^{*}$}  \\
        d^{*}\,, & \mbox{otherwise}
    \end{cases}
    \\
    M^{*}X & = & MX \cap D^{*}
    \\
M^{*}A  &=&
    \big((MA \cap \pow(D^{*})) \setminus (\{M^{*}X: X \in \mathcal{V}'_{1}\} \cup \pow_{\leq l}(\{M^{*}x : x \in
    \mathcal{V}'_{0}\}))\big)
    \\
    && \;\;\cup \big(\{M^{*}X: X \in \mathcal{V}'_{1},~MX \in MA\} \cup (\pow_{\leq l}(\{M^{*}x : x \in \mathcal{V}'_{0}\}) \cap
    MA)\big)\,.
\end{eqnarray*}
For ease of notation, we will often omit the reference to the
element $d^{*} \in D^{*}$ and write simply\/ $\Rel(\model, D^{*},
\mathcal{V}'_{0}, \mathcal{V}'_{1}, l)$ in place of\/ $\Rel(\model, D^{*},
d^{*}, \mathcal{V}'_{0}, \mathcal{V}'_{1}, l)$.
\eod
\end{definition}

Our goal is to show that any given satisfiable $\TLQSTR$-formula $\psi$
is satisfied by a small model of the form $\Rel(\model, D^{*},
\mathcal{V}'_{0}, \mathcal{V}'_{1}, l)$, where $\model =(D,M)$ is a model of $\psi$ and 
$D^{*}$ is a suitable subset of $D$ of bounded finite size.

At first, we state a slightly stronger result for
$\TLQSTR$-formulae which are propositional combinations of 
quantifier-free atomic formulae of levels 0 and 1.

%of the following types
%\begin{itemize}
%\item [] $x = y\,, \qquad x \in X\,, \qquad  \{x_1,\ldots,x_k\} = X\,, \qquad  \{x_1,\ldots,x_k\} \in A\,,$
%\item [] $X = Y \,, \qquad X \in A\,,$
%\end{itemize}
%with $x,y,x_1,\ldots,x_k \in \mathcal{V}_{0}$, $X,Y \in \mathcal{V}_{1}$, and $A \in
%\mathcal{V}_{2}$.

\begin{lemma}
\label{wasLe_basic} %Lemma 3.2
Let $\model=(D,M)$ and $\modelStar = \Rel(\model, D^{*}, d^*,
\mathcal{V}'_{0}, \mathcal{V}'_{1}, l)$ be, respectively, a $\TLQST$-interpretation and the relativized interpretation of $\model$ with respect to $D^{*} \subseteq
D$, $d^{*} \in D^{*}$, $\mathcal{V}_{0}' \subseteq
\mathcal{V}_{0}$,  $\mathcal{V}_{1}' \subseteq
\mathcal{V}_{1}$, and $l >0$.  Furthermore, let $\psi_{0}$ be a level $0$
quantifier-free atomic formula of the form $x = y$ or $x \in X$, with
$x,y \in \mathcal{V}_{0}$ and $X \in \mathcal{V}_{1}$, let $\psi_0'$ be a level 0 quantifier-free atomic formula of the form $\{x_1,\ldots,x_k\} = X$ or $\{x_1,\ldots,x_k\} \in A$, with $x_1,\ldots,x_k \in \mathcal{V}_{0}$, $X \in \mathcal{V}_{1}$, $A \in \mathcal{V}_{2}$, $k \leq l$, and let $\psi_1$ 
be a level $1$ quantifier-free atomic formula of the form
$X = Y$ or $X \in A$, with $X,Y \in \mathcal{V}_{1}'$, and $A \in
\mathcal{V}_{2}$.  Then we have:
\begin{itemize}
\item[(a)] if $Mx \in D^{*}$, for every $x\in \mathcal{V}_{0}$ in
$\psi_{0}$, then $\model \models \psi_{0}$ iff $\modelStar \models 
\psi_{0}$;

\item[(b)] if (b1) $Mx \in D^{*}$, for every $x\in \mathcal{V}_{0}$ in $\psi_{0}$, (b2) $M^* X = MX$, if $|MX| \leq l$ and $|M^*X|>l$ otherwise, for every $X \in  \mathcal{V}_{1}'$, and (b3) $M^* X = MX$, for every $X$ occurring in $\psi_{0}'$ such that $X \in  \mathcal{V}_{1} \setminus \mathcal{V}_{1}'$, then 
  $\model \models \psi_{0}'$ iff $\modelStar \models 
\psi_{0}'$;
\item[(c)] if (c1) $M^* X = MX$, if $|MX| \leq l$ and $|M^*X|>l$ otherwise, for every $X \in  \mathcal{V}_{1}'$,  and (c2) $(MX \mathop{\Delta} MY) \cap D^{*} \neq \emptyset$,\footnote{We
recall that $\mathop{\Delta}$ denotes the symmetric difference operator defined
by $s \mathop{\Delta} t = (s \setminus t) \cup (t \setminus s)$.} for all $X,Y
\in \mathcal{V}_{1}'$ such that $MX \neq MY$, then
$\model \models \psi_{1}$ iff $\modelStar \models 
\psi_{1}$.

\end{itemize}
\end{lemma}
\begin{proof}
Let us prove case (a) first. Assume $\psi_0 = x \in X$. $\model \models x \in X$ if and only if $Mx \in MX$. Since $Mx \in D^{*}$, by Definition \ref{relintrp}, $Mx = M^*x$ and thus  $Mx \in MX$ if and only if $M^*x \in MX$. Since $M^*x  \in D^{*}$,  $M^*x \in MX$ if and only if $M^*x \in MX\cap D^{*}$. Thus, by Definition \ref{relintrp},  $M^*x \in MX\cap D^{*}$ if and only if $M^*x \in M^*X$, and finally $M^*x \in M^*X$ if and only if 
$\modelStar \models  x \in X$, as we wished to prove. Next, let $\psi_0 = x = y$. $\model \models x = y$ if and only if $Mx = My$.  Since $Mx, My \in D^{*}$, by Definition \ref{relintrp}, $Mx = M^*x$ and $My = M^*y$ and thus  $Mx = My$ if and only if $M^*x = M^*y$. Finally $M^*x = M^*y$ if and only if $\modelStar \models  x = y$, and the thesis follows. 

For what concerns case (b), let us assume first that $\psi_0' = \{x_1,\ldots,x_k\} = X$, with $X \in \mathcal{V}_{1}'$. If $\model \models \{x_1,\ldots,x_k\} = X$, then $\{Mx_1,\ldots,Mx_k\} = MX$ and, since  $k \leq l$, $|MX| \leq l$ and therefore $M^*X = MX$. Moreover $Mx_1,\ldots,Mx_k \in D^{*}$ and thus, by Definition \ref{relintrp}, $Mx_i = M^*x_i$, for $i=1,\ldots,k$. Thus, if $\{Mx_1,\ldots,Mx_k\} = MX$, it holds that $\{M^*x_1,\ldots,M^*x_k\} = M^*X$, and finally that $\modelStar \models \{x_1,\ldots,x_k\} = X$, as we wished to prove. Conversely, assume that  $\model \not\models \{x_1,\ldots,x_k\} = X$. Then $\{Mx_1,\ldots,Mx_k\} \neq MX$. If $|MX| \leq l$, $M^*X = MX$, moreover, reasoning as above, $Mx_i = M^*x_i$, for $i=1,\ldots,k$. Thus, if  $\{Mx_1,\ldots,Mx_k\} \neq MX$, it holds that $\{M^*x_1,\ldots,M^*x_k\} \neq M^*X$, hence $\modelStar \not\models \{x_1,\ldots,x_k\} = X$ and the thesis follows.  Finally, if  $|MX| > l$, $|M^*X| > l$ and thus $|M^*X| > k$. As a consequence, it follows that $\{M^*x_1,\ldots,M^*x_k\} \neq M^*X$ and thus $\modelStar \not\models \{x_1,\ldots,x_k\} = X$, as we wished to prove. 
Next, assume that $\psi_0' = \{x_1,\ldots,x_k\} = X$, with $X \in \mathcal{V}_{1} \setminus  \mathcal{V}_{1}'$. Since $M^*X = MX$ and $Mx_i = M^*x_i$, for $i=1,\ldots,k$, $\model \models \{x_1,\ldots,x_k\} = X$ if and only if 
$\{Mx_1,\ldots,Mx_k\} = MX$ if and only if $\{M^*x_1,\ldots,M^*x_k\} = M^*X$ if and only if $\modelStar \models \{x_1,\ldots,x_k\} = X$. Hence, even in this case the thesis holds.

Finally, let $\psi_0' = \{x_1,\ldots,x_k\} \in A$. If $\model \models \{x_1,\ldots,x_k\} \in A$, then $\{Mx_1,\ldots,Mx_k\} \in MA$. In order to show that $\modelStar  \models \{x_1,\ldots,x_k\} \in A$,  we have to prove that $\{M^*x_1,\ldots,M^*x_k\} \in M^*A$. 

Since $Mx_1,\ldots,Mx_k \in D^{*}$, by Definition \ref{relintrp}, $Mx_i = M^*x_i$, for $i=1,\ldots,k$.  Thus
$\{M^*x_1,\ldots,M^*x_k\} = \{Mx_1,\ldots,Mx_k\}$ and 
$\{M^*x_1,\ldots,M^*x_k\} \in MA$. We may have that 
$\{M^*x_1,\ldots,M^*x_k\} \notin M^*A$ only in one of the following two cases. The first case is: $\{M^*x_1,\ldots,M^*x_k\} \in \pow_{\leq l}(\{M^{*}x : x \in
    \mathcal{V}'_{0}\})$ and $\{M^*x_1,\ldots,M^*x_k\} \notin MA$. This cannot occur because in fact $\{M^*x_1,\ldots,M^*x_k\} \in MA$. The other case to be considered is  $\{M^*x_1,\ldots,M^*x_k\} = M^*X$ with $MX \notin MA$, for some $X \in \mathcal{V}_{1}'$. If $\{M^*x_1,\ldots,M^*x_k\} = M^*X$, for some $X \in \mathcal{V}'_{1}$, then $|M^*X| \leq l$ and, therefore, $M^*X = MX$. Since $MX = \{M^*x_1,\ldots,M^*x_k\}$, the assumption $MX \notin MA$ contradicts the hypothesis that  $\{M^*x_1,\ldots,M^*x_k\} \in MA$. 
 Hence, 
 %the only possible cases for having $\{M^*x_1,\ldots,M^*x_k\} \notin M^*A$ cannot occur, 
 we must admit that if $\{Mx_1,\ldots,Mx_k\} \in MA$, then $\{M^*x_1,\ldots,M^*x_k\} \in M^*A$. 

 On the other hand, if $\model \not\models  \{x_1,\ldots,x_k\} \in A$, then $\{Mx_1,\ldots,Mx_k\} \notin MA$. Assume, by contradiction, that 
$\{M^*x_1,\ldots,M^*x_k\} \in M^*A$. Since $M^*x_i = Mx_i$, for $i=1,\ldots,k$, $\{M^*x_1,\ldots,M^*x_k\} = \{Mx_1,\ldots,Mx_k\}$ and thus, $\{M^*x_1,\ldots,M^*x_k\} \in M^*A$ only in the case $\{M^*x_1,\ldots,M^*x_k\} = M^*Z$, for some $Z \in \mathcal{V}_{1}'$ such that $MZ \in MA$. Since $|M^*Z| \leq l$, it holds that $MZ = M^*Z$ and thus $M^*Z \in MA$, and since $M^*Z = \{M^*x_1,\ldots,M^*x_k\} = \{Mx_1,\ldots,Mx_k\}$, we have $\{Mx_1,\ldots,Mx_k\} \in MA$, absurd. 

Finally, let us prove case (c). Let $\psi_1 = X= Y$. If $\model \models X = Y$, then $MX = MY$. Thus $MX \cap D^* = MY \cap D^*$ and,  by Definition \ref{relintrp}, $M^*X = M^*Y$. Since  $M^*X = M^*Y$, it immediately follows that $\modelStar \models X=Y$. On the other hand, if $\model \not\models X = Y$, then $MX \neq MY$. Thus $(MX \mathop{\Delta} MY) \cap D^{*} \neq \emptyset$ and, consequently, $MX \cap D^* \neq MY \cap D^*$. If $MX \cap D^* \neq MY \cap D^*$, by Definition \ref{relintrp}, $M^*X \neq M^*Y$ and finally $\modelStar \not\models X=Y$. Next, let us assume that $\psi_1 = X \in A$.
%, with $X \in
%\mathcal{V}_{1}'$ and $A \in \mathcal{V}_{2}$, and leave the remaining
%cases to the reader, as they are routine.

If $MX \in MA$, then $M^{*}X \in M^{*}A$ holds trivially.  On the
other hand, if $MX \notin MA$, but $M^{*}X \in M^{*}A$, then either 
$M^*X \in (\pow_{\leq l}(\{M^{*}x : x \in \mathcal{V}'_{0}\}) \cap
    MA)$ or $M^{*}X = M^{*}Z$, for some $Z \in
\mathcal{V}_{1}'$ such that $MZ \in MA$. In the first case, since $|M^*X| \leq l$, by (c1) it holds that $M^* X= MX$ and thus $MX \in MA$, absurd. 
In the other case, since $MZ \in MA$ it holds that $MX \neq
MZ$, and thus, by (c2), $(MX \mathop{\Delta} MZ) \cap D^{*}
\neq \emptyset$. The latter implies $M^{*}X \neq M^{*}Z$, a
contradiction. 
% By propositional logic, it is enough to show that
% \begin{equation}
%     \label{firstStepZero}
%     \model \models \psi_{0} \qquad \text{iff} \qquad \modelStar
%     \models \psi_{0}\,,
% \end{equation}
% for every atomic formula $\psi_{0}$ occurring in $\psi$.
\end{proof}

%The interested reader can find the proof of Lemma~\ref{wasLe_basic} in \cite{CanNicArxiv15}.
By propositional logic, Lemma~\ref{wasLe_basic} implies at once the 
following corollary.

\begin{corollary}
    \label{cor_basic} %Corollary 3.3
Let $\model=(D,M)$ and $\modelStar = \Rel(\model, D^{*}, d^*,
\mathcal{V}'_{0}, \mathcal{V}'_{1}, l)$ be, respectively, a $\TLQST$-interpretation and the relativized interpretation of $\model$ with respect to $D^{*} \subseteq
D$, $d^{*} \in D^{*}$, $\mathcal{V}_{0}' \subseteq
\mathcal{V}_{0}$,  $\mathcal{V}_{1}' \subseteq
\mathcal{V}_{1}$, and $l >0$.  Furthermore, let $\psi$ be a propositional
combination of quantifier-free atomic formulae of the types
$$x = y\,, \;\; x \in X\,,\;\; \{x_1,\ldots,x_k\} = X\,, \;\;\{x_1,\ldots,x_k\} \in A\,, \;\; X = Y \,,\;\; X \in A$$
%
%%\begin{itemize}
%%\item [] $x = y\,, \qquad x \in X\,, \qquad  \{x_1,\ldots,x_k\} = X\,, \qquad  \{x_1,\ldots,x_k\} \in A\,,$
%%\item [] $X = Y \,, \qquad X \in A\,,$
%%\end{itemize}
%\item [] $x = y\,, \qquad x \in X\,, \qquad  \{x_1,\ldots,x_k\} = X\,, \qquad  \{x_1,\ldots,x_k\} \in A\,,$
%\item [] $X = Y \,, \qquad X \in A\,,$
%\end{itemize}
%
such that
\begin{itemize}
    \item  $Mx \in D^{*}$, for every level $0$ variable $x$ in
    $\psi$;
   % \item $x \in \mathcal{V}'_{0}$, for every level $0$ variable $x$ in $\psi$;
   \item $k \leq l$;    
    \item  $X \in \mathcal{V}_{1}'$, for every level $1$
    variable $X$ in quantifier-free atomic formulae of level $1$ (namely of the form $X = Y$ or $X \in A$) occurring in $\psi$;
    \item $M^*X = MX$ if $|MX| \leq l$ and $|M^*X| > l$, otherwise, for every level $1$
    variable $X\in\mathcal{V}_{1}'$;
    \item  $(MX \mathop{\Delta} MY) \cap D^{*} \neq \emptyset$, for all
    $X,Y \in \mathcal{V}_{1}'$ such that $MX \neq MY$.
    \item $M^*X = MX$, for every  level $1$ variable $X\in\mathcal{V}_{1} \setminus \mathcal{V}_{1}'$ occurring in $\psi$. 
\end{itemize}
Then $\model \models \psi$ if and only if $\modelStar \models \psi$. 
%\[
%\model \models \psi \quad \text{iff} \quad \modelStar \models 
%\psi\,. \xqedhere{4.085cm}{\rule{.2cm}{.2cm}}
%\]
\end{corollary}

The preceding corollary yields at once a small model property for the 
collection \TLSZ of propositional combinations of
quantifier-free atomic formulae of the types
%$$x = y\,, x \in X\,, \{x_1,\ldots,x_k\} = X\,, \{x_1,\ldots,x_k\} \in A\,, X = Y \,, X \in A.$$
$$x = y\,, \;\; x \in X\,,\;\; \{x_1,\ldots,x_k\} = X\,, \;\;\{x_1,\ldots,x_k\} \in A\,, \;\; X = Y \,,\;\; X \in A$$
%\begin{itemize}
%\item [] $x = y\,, \qquad x \in X\,, \qquad  \{x_1,\ldots,x_k\} = X\,, \qquad  \{x_1,\ldots,x_k\} \in A\,,$
%\item [] $X = Y \,, \qquad X \in A\,.$
%\end{itemize}
Indeed, let $\psi$ be a satisfiable \TLSZ-formula and let
$\model=(D,M)$ be a model for it and let $l$ be the maximal length of  finite enumerations $\{x_1,\ldots,x_k\}$ occurring in $\psi$.  Let $\mathcal{V}^{\psi}_{0}$ and
$\mathcal{V}^{\psi}_{1}$ be respectively the collections of variables
of sort 0 and of sort 1 occurring in $\psi$.  
\begin{itemize}
\item For each pair of
variables $X,Y \in \mathcal{V}^{\psi}_{1}$ such that $MX \neq MY$, let
us select an element $d_{XY} \in MX \mathop{\Delta} MY$;
\item construct a set $D_1$ such that $|J \cap D_1| \geq \min(l+1,|J|)$, for every $J \in \{MX : X \in \mathcal{V}^{\psi}_{1}\}$. 
\end{itemize}
 Then put
$
D^{*} = \{Mx : x \in \mathcal{V}^{\psi}_{0}\} \cup (\{d_{XY} : X,Y \in
\mathcal{V}^{\psi}_{1}\,,~ MX \neq MY\} \cup D_1)\,.$
% Observe that $|D^{*}| \leq |\mathcal{V}^{\psi}_{0}| +
% |\mathcal{V}^{\psi}_{1}|^{2}$.  
Also, let $d^{*}$ be an arbitrarily chosen element of $D^{*}$.  Then,
from Corollary~\ref{cor_basic} it follows that the relativized
interpretation $\modelStar = \Rel(\model, D^{*}, d^{*}, \mathcal{V}^{\psi}_{0},
\mathcal{V}^{\psi}_{1}, l)$ is a \emph{small} model for $\psi$, as
$|D^{*}| \leq |\mathcal{V}^{\psi}_{0}| + (l+1)|\mathcal{V}^{\psi}_{1}| +
|\mathcal{V}^{\psi}_{1}|^{2}$.  In fact, it can be shown that the
elements $d_{XY}$ in the symmetric differences $MX \mathop{\Delta} MY$
can be selected in such a way that $|D^{*}| < |\mathcal{V}^{\psi}_{0}|
+ (l+2)|\mathcal{V}^{\psi}_{1}|$ holds
% \[
% |\{d_{XY} : X,Y \in \mathcal{V}^{\psi}_{1} \text{ such that }MX \neq
% MY\}| \leq |\mathcal{V}^{\psi}_{1}|
% \]
(see \cite{CanFer1995}). Summing up, the following result holds:
%%%%%%%%%%%%%%%%%%%%
\begin{lemma}[Small model property for \TLSZ-formulae]
Let $\psi$ be a \TLSZ-formula, i.e., a propositional combination
of
quantifier-free atomic formulae of the following forms
$$x = y\,, \;\; x \in X\,,\;\; \{x_1,\ldots,x_k\} = X\,, \;\;\{x_1,\ldots,x_k\} \in A\,, \;\; X = Y \,,\;\; X \in A$$
%\begin{itemize}
%\item [] $x = y\,, \qquad x \in X\,, \qquad  \{x_1,\ldots,x_k\} = X\,, \qquad  \{x_1,\ldots,x_k\} \in A\,,$
%\item [] $X = Y \,, \qquad X \in A\,,$
%\end{itemize}
and let $\mathcal{V}^{\psi}_{0}$ and $\mathcal{V}^{\psi}_{1}$ be the
collections of variables of sort 0 and of sort 1 occurring in $\psi$,
respectively.  Then $\psi$ is satisfiable if and only if is satisfied
by a $\TLQST$-interpretation $\model=(D,M)$ such that $|D^{*}| < |\mathcal{V}^{\psi}_{0}|
+ (l+2)|\mathcal{V}^{\psi}_{1}|$ .   \qed
\end{lemma}
Since the $\TLQST$-interpretations over a bounded domain are
finitely many and can be effectively generated, the decidability of
the satisfiability problem for \TLSZ-formulae follows.
\subsection{Relativized interpretations and quantified atomic formulae}

To state the main results on quantified formulae, namely that the relativized interpretation $\modelStar = \Rel(\model,
D^{*}, d^*, \mathcal{V}_{0}', \mathcal{V}_{1}', l)$ of a model $\model=(D,M)$ for a
purely universal $\TLQSTR$-formula $\psi$ of level $0$ or $1$ also
satisfies $\psi$ under suitable conditions on $D^{*}$, $\mathcal{V}_{0}'\subseteq \mathcal{V}_{0}$, $\mathcal{V}_{1}'\subseteq \mathcal{V}_{1}$, and $l$ (Lemmas \ref{quantifiedformOne} and \ref{quantifiedformTwo} below), it is convenient to introduce the following
abbreviations: 

%To show that the relativized interpretation $\modelStar = \Rel(\model,
%D^{*}, d^*, \mathcal{V}_{0}', \mathcal{V}_{1}', l)$ of a model $\model=(D,M)$ for a
%purely universal $\TLQSTR$-formula $\psi$ of level $0$ or $1$ also
%satisfies $\psi$ (under suitable conditions on $D^{*}$, $\mathcal{V}_{0}'\subseteq \mathcal{V}_{0}$, $\mathcal{V}_{1}'\subseteq \mathcal{V}_{1}$, and $l$), it is convenient to introduce the following
%abbreviations:
% 
\begin{eqnarray*}
   % \cal M}^{z}   & = & \model[z_{1}/u_{1},\ldots,z_{n}/u_{n}]  \\
    \modelzs & \defAs & \Rel(\modelz,D^{*},d^*,\mathcal{V}'_{0},\mathcal{V}'_{1},l)  \\
    \modelsz & \defAs & \modelStar[z_{1}/u_{1},\ldots,z_{n}/u_{n}] \\
    %\model^{Z}   & = & \model[Z_{1}/U_{1},\ldots,Z_{m}/U_{m}]  \\
    %\model^{*,z} & = & \model^{*}[z_{1}/u_{1},\ldots,z_{n}/u_{n}] \\
    \modelZs & \defAs & \Rel(\modelZ,D^{*},d^*,\mathcal{V}'_{0},\mathcal{V}'_{1}, l)  \\
    \modelsZ & \defAs & 
    \modelStar[Z_{1}/U_{1},\ldots,Z_{m}/U_{m}]\,,
    %\model^{Z,z} & = & \model^{Z}[z_{1}/u_{1},\ldots,z_{n}/u_{n}]
%    \\
%   \model^{Z,z,*} & = & \Rel(\model^{Z,z},D^{*},\mathcal{V}'_{1} \cup
%                          \{Z_{1},\ldots,Z_{m}\}) \\
%   \model^{Z,*,z} & = & \model^{Z,*}[z_{1}/u_{1},\ldots,z_{n}/u_{n}]
\end{eqnarray*}
with $z_{1},\ldots,z_{n} \in \mathcal{V}_{0} \setminus \mathcal{V}'_{0}$, $Z_{1},\ldots,Z_{m}
\in \mathcal{V}_{1} \setminus \mathcal{V}'_{1}$, $u_{1},\ldots,u_{n}
\in D$, $U_{1},\ldots,U_{m} \subseteq D$.

% RIPRISTINARE IN VERSIONE LUNGA
When $u_{1},\ldots,u_{n} \in D^{*}$, the
$\TLQST$-interpretations $\modelzs$ and $\modelsz$
coincide, as stated in the following lemma, whose proof is routine 
and is omitted for brevity.

\begin{lemma}
\label{le_M*zMz*} %Lemma 3.4
Let $\model=(D,M)$ be a $\TLQST$-interpretation, $D^{*} \subseteq
D$, $u_{1},\ldots,u_{n} \in D^{*}$, and 
$z_{1},\ldots,z_{n} \in \mathcal{V}_{0}\setminus \mathcal{V}'_{0}$. Then
the $\TLQST$-interpretations 
$\modelzs$ and $\modelsz$ coincide.\qed
\end{lemma}

Likewise, under some conditions, the $\TLQST$-interpretations
$\modelZs$ and $\modelsZ$ coincide too, as stated in the following
lemma.

% The following lemmas provide useful technical results to be employed
% in the proof of Theorem \ref{correctness} below.  In particular,
% Lemmas \ref{le_M*zMz*} and \ref{le_eqM*ZMZ*}, which are simply stated
% without proof, are used to prove Lemma \ref{quantifiedform}.

%{\bf Proof.} Since  $u_{1},\ldots,u_{n} \in D^{*}$, property (i)
%follows immediately.
%
%Concerning (ii), let $X \in \mathcal{V}_{1}$. Then $M^{*,z}X = M^{*}X =
%MX \cap D^{*} = M^{z}X \cap D^{*} = M^{z,*}X$. \qed

\begin{lemma}
\label{le_eqM*ZMZ*} %Lemma 3.5
Let $\model=(D,M)$ be a $\TLQST$-interpretation, $D^{*} \subseteq
D$, $\mathcal{V}'_{1} \subseteq \mathcal{V}_{1}$, $Z_{1},\ldots,Z_{m}
\in \mathcal{V}_{1} \setminus \mathcal{V}'_{1}$, and
$U_{1},\ldots,U_{m} \in \pow(D^*) \setminus \{M^{*}X : X \in
\mathcal{V}'_{1}\}$.
Then the $\TLQST$-interpretations $\modelZs$ and
$\modelsZ$ coincide. \qed
\end{lemma}

We are now ready to prove the main result of the present section,
namely that if $\model=(D,M)$ satisfies a purely universal
$\TLQSTR$-formula $\psi$ of level $0$ or $1$, then, under suitable
conditions, the relativized interpretation $\modelStar = \Rel(\model,
D^{*}, d^*, \mathcal{V}_{0}', \mathcal{V}_{1}', l)$ of $\model$ satisfies $\psi$ too.
This will be done in the following two lemmas.

\begin{lemma}
\label{quantifiedformOne} %Lemma 3.6
Let $\model=(D,M)$ be a $\TLQST$-interpretation, $D^{*} \subseteq D$, $d^{*} \in D^{*}$, $\mathcal{V}'_{0} \subseteq \mathcal{V}_{0}$, $\mathcal{V}'_{1} \subseteq \mathcal{V}_{1}$, $l >0$,
and let $\modelStar = \Rel(\model, D^{*}, d^{*},
\mathcal{V}'_{0}, \mathcal{V}'_{1}, l)$ be such that 
$M^* X = MX$, if $|MX| \leq l$ and $|M^*X|>l$ otherwise, for every $X \in  \mathcal{V}_{1}'$.  Furthermore, let $(\forall z_1) \ldots (\forall
z_n) \varphi_0$ be a purely universal $\TLQSTR$-formula of level $0$
such that 
\begin{itemize}
\item [(i)] $Mx \in D^*$, for every $x \in \mathcal{V}_{0}$ occurring
free in it;
\item [(ii)] Each occurrence of finite enumeration $\{x_1,\ldots,x_k\}$ in $\psi$, with $x_i \in \mathcal{V}_{0}$, for every $i \in \{1,\ldots,k\}$, is such that $k \leq l$;
\item [(iii)] $\{z_1,\ldots,z_n\} \in \mathcal{V}_{0} \setminus \mathcal{V}_{0}'$;
\item [(iv)] $M^*X = MX$, for every variable $X$ of level $1$ in $\psi$ such that $X \in \mathcal{V}_{1}\setminus \mathcal{V}_{1}'$. 
\end{itemize}
 Then 
\[
\model \models (\forall z_1) \ldots (\forall z_n) \varphi_0 \quad
\Longrightarrow \quad \modelStar \models (\forall z_1) \ldots (\forall
z_n) \varphi_0\,.
\]
\end{lemma}
\begin{proof}
Let $\model$ and $\modelStar$ be as in the lemma, and assume that
$\model \models (\forall z_1) \ldots (\forall z_n) \varphi_0$ whereas
$\modelStar \not\models (\forall z_1) \ldots (\forall z_n) \varphi_0$.
Then there must exist $u_1,\ldots,u_n \in D^*$ such that
$\modelStar[z_{1}/u_{1},\ldots,z_{n}/u_{n}] \not\models \varphi_0$,
i.e., $\modelsz\not\models \varphi_0$.  Since, by (iii), $\{z_1,\ldots,z_n\} \in \mathcal{V}_{0} \setminus \mathcal{V}_{0}'$, by
Lemma~\ref{le_M*zMz*}, $\modelzs \not\models \varphi_0$.  

By (i) and by the definition of $\model^z$, it is easy to see that
$\assignz x \in D^{*}$, for every $x \in
\mathcal{V}_{0}$ occurring in $\varphi_{0}$. Moreover, by (ii) each occurrence of finite enumeration $\{x_1,\ldots,x_k\}$ in $\varphi_{0}$, with $x_i \in \mathcal{V}_{0}$, for every $i \in \{1,\ldots,k\}$, is such that $k \leq l$. Finally, since $M^z X = MX$ and $M^{z,*} X= M^{*,z} X = M^* X$, for every variable $X \in \mathcal{V}_{1}$, it can be checked that
\begin{itemize}
\item $M^{z,*}X = M^zX$, for every variable $X$ of level $1$ occurring in $\psi$ such that $X \in \mathcal{V}_{1}\setminus \mathcal{V}_{1}'$ (by (iv)), and 
\item  $M^{z,*} X = M^z X$, if $|M^zX| \leq l$ and $|M^{z,*}X|>l$ otherwise, for every $X \in  \mathcal{V}_{1}'$ (because $M^* X = MX$, if $|MX| \leq l$ and $|M^*X|>l$ otherwise, for every $X \in  \mathcal{V}_{1}'$). 
\end{itemize}
Thus, by Lemma
\ref{wasLe_basic} (a) and (b) we have $\modelz \not\models \varphi_{0}$,
which yields $\model \not\models (\forall z_1) \ldots (\forall z_n)
\varphi_0$, a contradiction.
\end{proof}

% \newpage
\begin{lemma}
\label{quantifiedformTwo} %Lemma 3.7
Let $\model=(D,M)$ be a $\TLQST$-interpretation, $D^{*} \subseteq
D$, $d^{*} \in D^{*}$, $\mathcal{V}'_{0} \subseteq \mathcal{V}_{0}$, $\mathcal{V}'_{1} \subseteq \mathcal{V}_{1}$, $l >0$,
$\modelStar = \Rel(\model, D^{*}, d^{*}, \mathcal{V}'_{0}, \mathcal{V}'_{1},l)$, and let
$(\forall Z_1) \ldots (\forall Z_m) \varphi_1$ be a purely universal
$\TLQSTR$-formula of level $1$ such that 
\begin{enumerate}[label=({\roman*})]
   \item\label{new_i} $Z_{1},\ldots,Z_{m} \notin \mathcal{V}'_{1}$;
   
   \item\label{old_ii} %was (ii)
   $X \in \mathcal{V}'_1$, for every variable $X \in
   \mathcal{V}_1$ occurring free in $(\forall Z_1) \ldots (\forall
   Z_m) \varphi_1$;
   
   \item\label{old_i} %was (i)
   $Mx \in D^*$, for every $x \in \mathcal{V}_{0}$ occurring free in
   $(\forall Z_1) \ldots (\forall Z_m) \varphi_1$;
   
   \item\label{newtup_i} 
   $M^* X = MX$, if $|MX| \leq l$ and $|M^*X|>l$ otherwise, for every $X \in  \mathcal{V}_{1}'$;
      
   \item\label{old_iii} %was (iii)
   $(MX \mathop{\Delta} MY) \cap D^{*} \neq \emptyset$, for all $X,Y
   \in \mathcal{V}'_1$ such that $MX \neq MY$;
   
   \item\label{newtup_ii} 
   each occurrence of finite enumeration $\{x_1,\ldots,x_k\}$ in $\varphi_1$, with $x_i \in \mathcal{V}_{0}$, for every $i \in \{1,\ldots,k\}$, is such that $k \leq l$; 
      
   \item\label{old_iv} %was (iv)
   for every purely universal formula $(\forall z_1) \ldots (\forall
   z_n) \varphi_0$ of level $0$ occurring in $\varphi_{1}$ and
   variables $X_{1}, \ldots,X_{m} \in \mathcal{V}_{1}'$ such that
   $\model \not\models ((\forall z_1) \ldots (\forall z_n)
   \varphi_0)_{X_1,\ldots,X_m}^{Z_1,\ldots,Z_m}$, there are
   $u_1,\ldots,u_n \in D^*$ such that
   $\model[z_1/u_1,\ldots,z_n/u_n]\not\models
   (\varphi_0)_{X_1,\ldots,X_m}^{Z_1,\ldots,Z_m}$;\footnote{Given a
   formula $\psi$ and variables $X_1,\ldots,X_m,Z_1,\ldots,Z_m$, by
   $\psi_{X_1,\ldots,X_m}^{Z_1,\ldots,Z_m}$ we mean the formula
   obtained by simultaneously substituting 
   $Z_1,\ldots,Z_m$ with $X_1,\ldots,X_m$ in $\psi$.}
 \item\label{newtup_iii}  
 for every purely universal formula $(\forall z_1) \ldots (\forall
   z_n) \varphi_0$ of level $0$ occurring in $\varphi_{1}$, $\{z_1,\ldots,z_n\} \in \mathcal{V}_{0} \setminus \mathcal{V}_{0}'$.
\end{enumerate}
Then
\[
\model \models (\forall Z_1) \ldots (\forall Z_m) \varphi_1 \quad
\Longrightarrow \quad \modelStar \models (\forall Z_1) \ldots (\forall
Z_m) \varphi_1\,.
\]
\end{lemma}
\begin{proof}
Let $\model$, $\modelStar$, and $(\forall Z_1) \ldots (\forall Z_m)
\varphi_1$ be as in the lemma, and assume that $\model \models
(\forall Z_1) \ldots (\forall Z_m) \varphi_1$ whereas $\modelStar
\not\models (\forall Z_1) \ldots (\forall Z_m) \varphi_1$.  Then there
must exist $U_1, \ldots, U_m \subseteq D^*$ such that
$\modelStar[Z_{1}/U_{1},\ldots,Z_{m}/U_{m}] \not\models \varphi_1$,
i.e.,
\begin{equation}
    \label{eq_MstarZ}
    \modelsZ\not\models\varphi_1\,.
\end{equation}

Without loss of generality, we may assume that there exists $0 \leq h
\leq m$ such that
\begin{itemize}
    \item $U_i = \assignStar X_i$, for $1 \leq i \leq h$, for some
    variables $X_1,\ldots,X_h$ in $\mathcal{V}_1'$, and
    
    \item $U_j \notin \{\assignStar X : X \in \mathcal{V}_1'\}$, for
    all $h+1 \leq j \leq m$.
\end{itemize}
    
% Without loss of generality, we may assume that $U_i = \assignStar
% X_i$, for $1 \leq i \leq k$ (with $0 \leq k \leq m$) for some variables
% $X_1,\ldots,X_k$ in $\mathcal{V}_1'$, and that $U_j \notin
% \{\assignStar X : X \in \mathcal{V}_1'\}$, for all $k+1 \leq j \leq
% m$.

Let $\bar{\varphi}_1 \defAs (\varphi_{1})^{Z_{1}\ldots
Z_{h}}_{X_{1}\ldots X_{h}}$ (i.e., $\bar{\varphi}_1$ is the formula
obtained by simultaneously substituting $Z_1,\ldots,Z_h$ with
$X_1,\ldots,X_h$ in $\varphi_1$) and let 
\[
\modelZm \defAs
\model[Z_{h+1}/U_{h+1},\ldots,Z_{m}/U_m]\,.
\]

Our plan is to show that 
\begin{equation}
    \label{modelZm}
    \modelZm \not\models \bar{\varphi}_1
\end{equation}
holds.  Then, since (\ref{modelZm}) readily implies
\begin{equation}
    \label{eq_five}
    \modelZp \not\models \varphi_1\,,
\end{equation}
where
\[
\modelZp \defAs
\model[Z_{1}/MX_{1},\ldots,Z_{h}/MX_{h},
Z_{h+1}/U_{h+1},\ldots,Z_{m}/U_{m}]\,,
\]
and (\ref{eq_five}) in its turn yields $\model \not\models (\forall
Z_1) \ldots (\forall Z_m) \varphi_1$, a contradiction would be
derived, proving that $\model \models (\forall Z_1) \ldots (\forall
Z_m) \varphi_1$ implies $\modelStar \models (\forall Z_1) \ldots
(\forall Z_m) \varphi_1$ (and hence completing the proof of the
lemma).

Thus, in what follows we will just show that (\ref{eq_MstarZ}) implies
(\ref{modelZm}).

To begin with, let $\modelsZm \defAs
\modelStar[Z_{h+1}/U_{h+1},\ldots,Z_{m}/U_m]$.  Plainly,
(\ref{eq_MstarZ}) implies at once $\modelsZm \not \models
\bar{\varphi}_1$.  Since, by hypothesis \ref{new_i} of the lemma and by
Lemma~\ref{le_eqM*ZMZ*}, $\modelsZm$ and $\modelZms$ coincide, so that
$\modelZms \not \models \bar{\varphi}_1$ holds, to prove
(\ref{modelZm}) it will be enough, by propositional logic, to show
that $\modelZms$ and $\modelZm$ coincide on all propositional
components\footnote{\label{myFootnote}By definition, a formula $\psi$
of $\TLQST$ is a propositional combination of certain atomic formulae
of level $0$, $1$, and $2$.  These are the \textsc{propositional
components} of $\psi$.} of $\bar{\varphi}_1$, which is what we do
next.

% To begin with, by Lemma~\ref{le_eqM*ZMZ*}, $\modelsZm$ and 
% $\modelZms$ coincide. Thus, it is sufficient to show that $\modelZms$ 
% and $\modelZm$ coincide on  all propositional
% components of $\bar{\varphi}_1$.

By hypotheses \ref{old_ii}, \ref{old_i}, \ref{newtup_i}, \ref{old_iii}, and \ref{newtup_ii} of
the lemma and by Lemma~\ref{wasLe_basic}, $\modelZms$ and $\modelZm$
coincide on all propositional components of $\bar{\varphi}_1$ of any
of the following types:
\begin{itemize}
    \item $x = y$, $x \in X$ ~(with $x,y \in \mathcal{V}_{0}$ and $X
    \in \mathcal{V}_{1}$), 
    \item $\{x_1,\ldots,x_k\} = X$,  $\{x_1,\ldots,x_k\} \in A$ ~(with $x_1,\ldots,x_k \in \mathcal{V}_{0}$, $X \in \mathcal{V}_{1}$, and  $A \in \mathcal{V}_{2}$), and
    \item $X = Y$, $X \in A$ ~(with $X,Y \in \mathcal{V}_{1}'$ and
    $A \in \mathcal{V}_{2}$).
\end{itemize}
Thus, to complete the proof, we are only left with showing that
$\modelZms$ and $\modelZm$ coincide also on the propositional
components of $\bar{\varphi}_1$ of the remaining types, namely those
of the form:
\begin{itemize}
    \item $Z_{j} = X$, $X = Z_{j}$, $Z_{j} \in A$\\ (with $X \in
    \mathcal{V}_{1}' \cup \{Z_{h+1},\ldots,Z_{m}\}$, $A \in \mathcal{V}_{2}$, and $h+1 \leq j \leq m$), and
    
    \item level $0$ purely universal formulae.
\end{itemize}

For propositional components of $\bar{\varphi}_1$ of type $Z_{j} = X$
(with $X \in \mathcal{V}_{1} \setminus \{Z_{1},\ldots,Z_{h}\}$ and $h+1 \leq j \leq m$), we have:
\[
\begin{array}{rcl}
\modelZms \models Z_{j} = X &\Longleftrightarrow &
U_{j} = \assignZm X \cap D^{*}\\
&\Longleftrightarrow & 
X \equiv Z_{i}, \text{ for some } i \in \{h+1, \ldots,m\}\\
&&\text{ such that } U_{i} = U_{j}\\
% &\Longleftrightarrow & \assignZm Z_{j} = \assignZm X\\
&\Longleftrightarrow & \modelZm \models Z_{j} = X\,.
\end{array}
\]
Analogously, for propositional components of $\bar{\varphi}_1$ of type
$X = Z_{j}$, with $X \in \mathcal{V}_{1}\setminus
\{Z_{1},\ldots,Z_{h}\}$ and $h+1 \leq j \leq m$.

For propositional components of $\bar{\varphi}_1$ of type $Z_{j} \in
A$ ~(with $A \in \mathcal{V}_{2}$ and $h+1
\leq j \leq m$), we have:
\[
\begin{array}{rcl}
\modelZms \models Z_{j} \in A &\Longleftrightarrow &
U_{j} \in \big((\assignZm A \cap \pow(D^{*})) \setminus 
(\{\assignZms X: X \in \mathcal{V}'_{1}\}\\
&& \qquad \quad ~~ {} \cup \; \pow_{\leq l}(\{\assignZms x : x \in \mathcal{V}'_{0}\}) \big)\\
&& \qquad \quad ~~ {} \cup \big(\{\assignZms X: X \in 
\mathcal{V}'_{1},~\assignZm X
\in \assignZm A\}\\
& & \qquad \quad ~~ {} \cup \; (\pow_{\leq l}(\{\assignZms x : x \in \mathcal{V}'_{0}\}) \cap \assignZm A)\big) \\
% &\Longleftrightarrow & U_{j} \in \assignZm A 
% \text{ ~or~ } 
% U_{j} = \assignZms X \text{, for some } X \in \mathcal{V}_{1}'\\
% && \text{such that } \assignZm X \in \assignZm A\\
&\Longleftrightarrow & U_{j} \in \assignZm A \hfill \text{{\small (since $U_j \notin
\{\assignStar X : X \in \mathcal{V}_1'\}$ }}\\
& & \hfill \text{{\small and $U_j \in \pow(D^{*})$)}}\\
&\Longleftrightarrow & \assignZm Z_{j} \in \assignZm A\\
&\Longleftrightarrow & \modelZm \models Z_{j} \in A\,.
\end{array}
\]

%For propositional components of $\bar{\varphi}_1$ of type $\{x_1,\ldots,x_k\} = Z_j$, with $k \leq l$ and $h+1 \leq j \leq m$: 
%\[
%\begin{array}{rcl}
%\modelZms \models \{x_1,\ldots,x_k\} = Z_{j} &\Longleftrightarrow & \{\assignZms x_1,\ldots, \assignZms x_k\} = U_{j} \\
%&\Longleftrightarrow & 
%\{\assignZm x_1,\ldots, \assignZm x_k\} = \assignZm  Z_{j}\\
%&\Longleftrightarrow & \modelZm \models \{x_1,\ldots,x_k\} = Z_{j}\,.
%\end{array}
%\]

Finally, let $(\forall z_1) \ldots (\forall z_n) \varphi_0$ be a
propositional component of $\varphi_{1}$ and let 
$\bar\varphi_0 \defAs (\varphi_0)_{X_1,\ldots,X_h}^{Z_1,\ldots,Z_h}$.
We show that 
\begin{equation}
\label{eqEquivalence}
\modelZms \models (\forall z_1) \ldots (\forall z_n) \bar\varphi_0
\quad \Longleftrightarrow \quad
\modelZm \models (\forall z_1) \ldots (\forall z_n) \bar\varphi_0\,.
\end{equation}

Let us first assume that 
\begin{equation}
    \label{initialAssumption}
    \modelZms \models (\forall z_1) \ldots (\forall z_n) \bar\varphi_0
\end{equation}
but, by way of contradiction, that 
\begin{equation}
    \label{contradictoryAssumption}
    \modelZm \not\models (\forall z_1) \ldots (\forall z_n)
    \bar\varphi_0\,.
\end{equation}
We will distinguish two cases, according to whether $h< m$ (i.e.,
$\{U_1, \ldots, U_{m}\} \not\subseteq \{\assignStar X : X \in
\mathcal{V}_1'\}$) or $h = m$ (i.e., $\{U_1, \ldots, U_{m}\} \subseteq
\{\assignStar X : X \in \mathcal{V}_1'\}$).

\smallskip

\noindent \textbf{Case $h < m$\,:} %
% From $\modelZm \not\models (\forall z_1) \ldots (\forall z_n)
% \bar\varphi_0$, it follows that there exist $u_{1},\ldots,u_{n} \in D$
% such that $\modelZm[z_{1}/u_{1},\ldots,z_{n}/u_{n}] \not\models
% \bar\varphi_0$, i.e., 
From $\modelZm \not\models (\forall z_1) \ldots (\forall z_n)
\bar\varphi_0$, it follows that there exist $u_{1},\ldots,u_{n} \in D$
such that $\modelZm[z_{1}/u_{1},\ldots,z_{n}/u_{n}] \not\models
\bar\varphi_0$. Let us put $\modelZmz \defAs \modelZm[z_{1}/u_{1},\ldots,z_{n}/u_{n}]$. Then 
we have
\begin{equation}
    \label{temp1}
    \modelZmz \models \neg
    \bar\varphi_0\,.
\end{equation}
Recalling that by definition of $\TLQSTR$-formulae (cf.\
Section~\ref{restrictionquant}) the formula $(\forall z_1) \ldots
(\forall z_n) \varphi_0$ must be linked to the variables
$Z_{1},\ldots,Z_{m}$, then we have
\[
\models \neg \varphi_0 \rightarrow \bigwedge_{i=1}^n \bigwedge_{j=1}^m  z_i \in Z_j
\]
(cf.\ condition (\ref{condition})), so that
\[
\models \Bigg(\neg \varphi_0 \rightarrow \bigwedge_{i=1}^n
\bigwedge_{j=1}^m z_i \in 
Z_j\Bigg)\!\rule[-3mm]{0pt}{8mm}_{X_{1},\ldots,X_{h}}^{Z_{1},\ldots,Z_{h}}\,,
\]
i.e.,
\[
\models \neg \bar\varphi_0 \rightarrow \bigwedge_{i=1}^n
\Bigg(\bigwedge_{j=1}^h z_i \in X_j \wedge \bigwedge_{j=h+1}^m z_i \in
Z_j\Bigg)\,.
\]
Thus, by (\ref{temp1}),
\[
\modelZmz \models \bigwedge_{i=1}^n \Bigg(\bigwedge_{j=1}^h z_i \in
X_j \wedge \bigwedge_{j=h+1}^m
  z_i \in Z_j\Bigg)\,,
\]
so that, for $i=1,\ldots,n$,
\[
\modelZmz \models z_i \in Z_m\,.
\]
Therefore, for  $i=1,\ldots,n$,
\begin{equation}
\label{eq_uiInDstar}
u_{i} = \assignZmz z_{i} \in \assignZmz Z_{m} = U_{m} \subseteq D^{*}\,. 
\end{equation}
In view of (\ref{eq_uiInDstar}) and by conditions \ref{old_ii}, \ref{old_i}, \ref{newtup_i}, \ref{old_iii}, \ref{newtup_ii}, and \ref{newtup_iii}  of
this lemma, we can apply Corollary~\ref{cor_basic} and Lemma \ref{le_M*zMz*} in the deductions
which follow:
\[
\begin{array}{rcl}
    \modelZmz \models \neg \bar\varphi_0 &\Longrightarrow &
    (\modelZm)^{\vec{z}} \models \neg \bar\varphi_0 \\   
    &\Longrightarrow &
    (\modelZm)^{\vec{z},*} \models \neg \bar\varphi_0 \hfill 
    \phantom{aaaaaa}\text{{\small (from
    (\ref{eq_uiInDstar}), conditions  \ref{old_ii}, \ref{old_i},  }} \\
    &&\hfill\text{{\small \ref{newtup_i}, \ref{old_iii}, and \ref{newtup_ii} of  the)}} \\[4pt] 
    &&\hfill\text{{\small present lemma, and 
    Corollary~\ref{cor_basic})}} \\[4pt] 
    &\Longrightarrow &
    (\modelZm)^{*,\vec{z}} \models \neg \bar\varphi_0
    \hfill\text{{\small (from (\ref{eq_uiInDstar}), condition \ref{newtup_iii},}} \\[4pt] & & \hfill\text{{\small  and Lemma \ref{le_M*zMz*})}}
    \\[4pt] &\Longrightarrow & \modelZms \not\models (\forall
    z_{1})\ldots(\forall z_{n}) \bar\varphi_0 \,.
\end{array}
\]
Hence $\modelZms \not\models (\forall z_{1})\ldots(\forall z_{n})
\bar\varphi_0$ holds, contradicting our initial assumption
(\ref{initialAssumption}) and therefore proving that the case $k < m$
can not arise.

\smallskip

\noindent \textbf{Case $h=m$\,:} When $h = m$, the interpretations
$\modelZm$ and $\modelZms$ are just $\model$ and $\modelStar$,
respectively.  Thus, our contradictory assumption
(\ref{contradictoryAssumption}) becomes $\model \not\models (\forall
z_{1})\ldots(\forall z_{n}) \bar\varphi_0$, which, by condition
\ref{old_iv} of the lemma, implies the existence of elements
\begin{equation}\label{uInDStar}
u_{1},\ldots,u_{n} \in D^{*}
\end{equation}
such that
$\model[z_{1}/u_{1},\ldots,z_{n}/u_{n}] \not\models \bar\varphi_0$,
i.e., $\modelz \not\models \bar\varphi_0$.  But,
\[
\begin{array}{rcll}
    \modelz \not\models \bar\varphi_0 & \Longrightarrow & \modelzs
    \not\models \bar\varphi_0 \hspace{2cm} \hfill \text{{\small (from (\ref{uInDStar}),
    conditions \ref{old_ii}, \ref{old_i}, \ref{newtup_i}, \ref{old_iii}, }}\\
    && \hfill\text{{\small and \ref{newtup_ii} of the lemma, and Corollary~\ref{cor_basic})}} \\[4pt]
    &\Longrightarrow&\modelsz \not\models \bar\varphi_0 \hfill
    \text{{\small (from (\ref{uInDStar}), condition \ref{newtup_iii}, and Lemma \ref{le_M*zMz*})}} \\[4pt]
    &\Longrightarrow&\modelStar\not\models (\forall
    z_{1})\ldots(\forall z_{n}) \bar\varphi_0 \,.
\end{array}
\]
Therefore $\modelZms\not\models (\forall z_{1})\ldots(\forall z_{n})
\bar\varphi_0$, which contradicts our assumption
(\ref{initialAssumption}).  Thus, even the current case $h = m$ can
not arise.  Since in any case we get a contradiction, we have the
following implication:
\[
\modelZms \models (\forall z_1) \ldots (\forall z_n) \bar\varphi_0
\quad \Longrightarrow \quad
\modelZm \models (\forall z_1) \ldots (\forall z_n) \bar\varphi_0\,.
\]
To complete the proof of (\ref{eqEquivalence}), we need to establish
also the converse implication.  But this follows at once, by observing
that if $\modelZm \models (\forall z_1) \ldots (\forall z_n)
\bar\varphi_0$, then by conditions \ref{old_i}, \ref{newtup_ii}, and \ref{newtup_iii} of the lemma and by
Lemma~\ref{quantifiedformOne} we have $\modelZms \models (\forall z_1)
\ldots (\forall z_n) \bar\varphi_0$.

This concludes the proof of the lemma.
% Thus, let us assume that
% $\modelZm \models (\forall z_1) \ldots (\forall z_n) \bar\varphi_0$.
% By condition \ref{old_i} of the lemma and by
% Lemma~\ref{quantifiedformOne}, we have at once $\modelZms \models
% (\forall z_1) \ldots (\forall z_n) \bar\varphi_0$, thus proving
% \[
% \modelZm \models (\forall z_1) \ldots (\forall z_n) \bar\varphi_0
% \quad \Longrightarrow \quad
% \modelZms \models (\forall z_1) \ldots (\forall z_n) \bar\varphi_0\,,
% \]
% which, together with the converse implication that we have already
% established, yields (\ref{eqEquivalence}), completing the proof
% of the lemma.
\end{proof}
%Proofs of Lemmas \ref{quantifiedformOne} and \ref{quantifiedformTwo} can be found in \cite{CanNicArxiv15}. 

\section{The satisfiability problem for
$\TLQSTR$-formulae}\label{satisfiability}
We will solve the satisfiability problem for $\TLQSTR$
, i.e., the
problem of establishing for any given formula of $\TLQSTR$ whether it
is satisfiable or not, 
as follows:
\begin{itemize}
    \item[(a)] firstly, we will reduce effectively the satisfiability
    problem for $\TLQSTR$-formulae to the same problem for normalized
    $\TLQSTR$-conjunctions (these will be defined precisely below);

    \item[(b)] secondly, we will prove that the collection of
    normalized $\TLQSTR$-conjunctions enjoys a small model property.
\end{itemize}
From (a) and (b), the solvability of the satisfiability problem for
$\TLQSTR$ will follow immediately.  
In fact, by further elaborating on
point (a), it could easily be shown that the whole collection of
$\TLQSTR$-formulae enjoys a small model property.

%We must stress, however, that such an approach allows only to
%establish the decidability of $\TLQSR$-formulae and does not suggest
%any practical decision procedure: this will be the subject of further
%research.

\subsection{Normalized $\TLQSTR$-conjunctions}\label{normal3LQS}
Let $\psi$ be a formula of $\TLQSTR$ and let $\psi_{\textit{DNF}}$ be a
disjunctive normal form of $\psi$.  We observe that the disjuncts of
$\psi_{\textit{DNF}}$ are conjunctions of $\TLQSTR$-literals, namely
quantifier-free atomic formulae of levels $0$ and $1$, or their
negations, and of purely universal formulae of levels $0$ and $1$, or
their negations, satisfying the linkedness condition
(\ref{condition}).  
% Plainly, $\psi$ is satisfiable if and only if at
% least one of the disjuncts of $\psi_{\textit{DNF}}$ is satisfiable.

By a suitable renaming of variables, we can assume that no bound
variable can occur in more than one quantifier in the same disjunct of
$\psi_{\textit{DNF}}$ and that no variable can have both bound and
free occurrences in the same disjunct.

% Thus, without loss of generality, we can just address the
% satisfiability problem for conjunctions $\psi'$ of $\TLQSR$-literals
% of level $0$ and $1$.  
% and that all atomic formulae within the scope of a quantifier
% must contain at least one bound variable.

Without disrupting satisfiability, we replace negative literals of the
form $\neg (\forall z_1) \ldots (\forall z_n) \varphi_0$ and $\neg
(\forall Z_1) \ldots (\forall Z_m) \varphi_1$ occurring in
$\psi_{\textit{DNF}}$ by their negated matrices $\neg \varphi_{0}$ and
$\neg \varphi_{1}$, respectively, since for any given
\TLQST-interpretation $\model = (D, M)$ one has $\model \models \neg
(\forall z_1) \ldots (\forall z_n) \varphi_0$ if and only if
$\model[z_1/u_1,\ldots , z_n/u_n] \models \neg \varphi_0$, for some
$u_1,\ldots ,u_n \in D$, and, likewise, $\model \models \neg (\forall
Z_1) \ldots (\forall Z_m) \varphi_1$ if and only if
$\model[Z_1/U_1,\ldots , Z_m/U_m] \models \neg \varphi_1$, for some
$U_1,\ldots , U_m \in \pow(D)$.  Then, if needed, we bring back the
resulting formula into disjunctive normal form, eliminate as above the
residual negative literals of the form $\neg (\forall z_1) \ldots
(\forall z_n) \varphi_0$ which might have been introduced by the
previous elimination of negative literals of the form $\neg (\forall
Z_1) \ldots (\forall Z_m) \varphi_1$ from $\psi_{\textit{DNF}}$, and
transform again the resulting formula in disjunctive normal form.  Let
$\psi'_{\textit{DNF}}$ be the formula so obtained.  Observe that all
the above steps preserve satisfiability, so that our initial formula
$\psi$ is satisfiable if so is $\psi'_{\textit{DNF}}$.  In addition,
the formula $\psi'_{\textit{DNF}}$ is satisfiable if and only if so is
at least one of its disjuncts.

It is an easy matter to check the each disjunct of
$\psi'_{\textit{DNF}}$ is a conjunction of $\TLQSTR$-literals of the following
types:
\[
\begin{array}[c]{cccc}
x=y\,, & \quad  x \in X\,,  & \quad \{x_1,\ldots,x_k\} = X\,,  & \quad \{x_1,\ldots,x_k\} \in A\,,\\
\neg (x = y)\,,& \quad \neg (x \in X)\,,  & \quad
\neg(\{x_1,\ldots,x_k\} = X)\,,  & \quad \neg(\{x_1,\ldots,x_k\} \in A)\,,\\
 X=Y\,, & \quad X \in A\,, & \quad\neg(X=Y)\,,  & \quad \neg(X \in A)\,, 
\end{array} \label{was1}\tag{I}
\]
where $x,y,x_1,\ldots,x_k \in \mathcal{V}_{0}$, $X,Y \in \mathcal{V}_{1}$, and $A 
\in \mathcal{V}_{2}$;
\[
(\forall z_1) \ldots (\forall z_n) \varphi_0\,, \label{was2}\tag{II}
\]
where $n > 0$ and $\varphi_0$ is a propositional combination of
quantifier-free level $0$ atoms; and
\[
(\forall Z_1) \ldots (\forall Z_m) \varphi_1\,, \label{was3}\tag{III}
\]
where $m > 0$ and $\varphi_1$ is a propositional combination of
quantifier-free atomic formulae of any level and of purely universal
formulae of level $0$, where the propositional components in
$\varphi_{1}$ of type $(\forall z_1) \ldots (\forall z_n) \varphi_0$
are linked to the bound variables $Z_1, \ldots , Z_m$.

% \begin{enumerate}[label=(\Roman*)]%[(I)]
% % \begin{itemize}
% \item\label{was1} 
% ~~$
% \begin{array}[t]{cccc}
% x=y\,, & x \in X\,, & X=Y\,, & X \in A\,,\\
% \neg (x = y)\,, & \neg (x \in X)\,, & \neg(X=Y)\,, & \neg(X \in A)\,,
% \end{array}
% $
% 
% where $x,y \in \mathcal{V}_{0}$, $X,Y \in \mathcal{V}_{1}$, and $A 
% \in \mathcal{V}_{2}$;
% %was [(1)]
% 
% \item\label{was2} ~~$(\forall z_1) \ldots (\forall
% z_n) \varphi_0$, 
% 
% where $n > 0$ and $\varphi_0$ is a propositional combination of
% quantifier-free level $0$ atoms;
% %was [(2)]
% 
% \item\label{was3} ~~$(\forall Z_1) \ldots
% (\forall Z_m) \varphi_1$, 
% 
% where $m > 0$ and $\varphi_1$ is a propositional combination of
% quantifier-free atomic formulae of any level and of purely universal
% formulae of level $0$, where the propositional components in
% $\varphi_{1}$ of type $(\forall z_1) \ldots (\forall z_n) \varphi_0$
% are linked to the bound variables $Z_1, \ldots , Z_m$.
% %was [(3)]
% % \end{itemize}
% \end{enumerate}
%     
We call such formulae \emph{normalized $\TLQSTR$-conjunctions}.

%\smallskip

The above discussion can then be summarized in the following lemma.
\begin{lemma}
    The satisfiability problem for $\TLQSTR$-formulae can be 
    effectively reduced to the satisfiability problem for 
    $\TLQSTR$-conjunctions.
\end{lemma}

%In what follows we show that normalized $\TLQSTR$-conjunctions
%enjoy a small model property, thereby solving their satisfiability
%problem.

% ELIMINATE COMMENT
%\comment{
\subsection{A small model property for normalized $\TLQSTR$-conjunctions}
\label{decisionproc}

Let $\psi$ be a normalized $\TLQSTR$-conjunction and assume
that $\model = (D, M)$ is a model for $\psi$.
We show how to construct, out of $\model$, a finite ``small''
$\TLQST$-interpretation $\modelStar = (D^*,\assignStar )$ which is
a model of $\psi$.  We proceed as follows.  First we outline a
procedure to build a nonempty finite universe $D^* \subseteq D$ whose
size depends solely on $\psi$ and can be computed {\em a priori}.
Then, following Definition \ref{relintrp}, we construct a relativized
$\TLQST$-interpretation $\modelStar = (D^*,\assignStar )$ with
respect to suitable collections $\mathcal{V}_{0}'$ and $\mathcal{V}_{1}'$ of variables, and to a positive number $l$, and
show that $\modelStar$ satisfies $\psi$.

\subsubsection{Construction of the universe $D^*$.}\label{ssseUniv}
Let $\mathcal{V}_0^{\psi}$, $\mathcal{V}_1^{\psi}$, and
$\mathcal{V}_2^{\psi}$ be the collections of the variables of sort
$0$, $1$, and $2$ occurring in $\psi$, respectivelyand, an let $l_{\psi}$ be smallest number such that $k \leq l_{\psi}$, for every finite enumeration $\{x_1,\ldots,x_k\}$ occurring in $\psi$.  We compute $D^*$
by means of the procedure below.

%$\SmallDomain(\psi, \model)$ illustrated in Figure
%\ref{domain} as follows.
%
%We will show that there exists a nonempty subset ${\cal U}^*
%\subseteq {\cal U}$, whose size depends solely on $\psi$ and can be
%computed {\em a priori}, such that the interpretation $\modelStar =
%({\cal U}^*, \assignStar )$ defined by
%\begin{itemize}
%\item[(a)] $\assignStar x=Mx$,
%\item[(b)] $\assignStar X=MX \cap {\cal U}^*$,
%\item[(c)] $\assignStar A=(MA \cap {\cal U}_1^*) \cup \mathcal{V}_{1,A}^*$,
%\end{itemize}
%where
%\begin{eqnarray*}
%{\cal U}_1^* & = & pow({\cal U}^*) \setminus \{MX \cap {\cal U}^* :
%                                               X \hbox{ in } \mathcal{V}_1 \} \\
%\mathcal{V}_{1,A}^* & = & \{MX \cap {\cal U}^* : X \hbox{ in } {\cal
%V}_1
%                \hbox{ and } MX \in MA \}.
%\end{eqnarray*}
%is a model for $\psi$. Such a model $\modelStar$ will be also called
%a {\em canonical model for} $\psi$.

Let $\psi_1, \ldots , \psi_h$ be the conjuncts of $\psi$ of the form
(\ref{was3}).  To each such conjunct $\psi_i \equiv (\forall Z_{i1})
\ldots (\forall Z_{im_i}) \varphi_i$, we associate the collection
$\varphi_{i1}, \ldots , \varphi_{i\ell_i}$ of the propositional
components of its matrix $\varphi_i$ 
%(cf.\ footnote \ref{myFootnote})
and call the variables $Z_{i1}, \ldots , Z_{im_i}$ the \emph{arguments
of} $\varphi_{i1}, \ldots , \varphi_{i\ell_i}$.  Then we put
\[
\Phi \defAs \{ \varphi_{ij} : 1 \leq i \leq h \hbox{ and } 1 \leq j \leq
\ell_i \}.
\]

By applying the procedure \emph{Distinguish} described in
\cite{CanFer1995} to the collection $\{MX : X \in
\mathcal{V}_1^{\psi}\}$, it is possible to construct a set $D_{0}$ 
such that
\begin{itemize}
    \item $MX \cap D_{0} \neq MY \cap D_{0}$, for all $X,Y \in
    \mathcal{V}_1^{\psi}$ such that $MX \neq MY$, and
    
    \item $|D_{0}| \leq |\mathcal{V}_1^{\psi}| -1$.
\end{itemize}

Next, we construct a set $D_1$ satisfying that $|J \cap D_1| \geq \min(l_{\psi} +1, |J|)$, for every $J \in \{MX : X \in \mathcal{V}_1^{\psi}\}$. Plainly, we
can assume that $|D_{1}| \leq (l_{\psi} +1)|\mathcal{V}_1^{\psi}|$.

%It is easy to check that $|D_{1}| \leq (l_{\psi}+1)|\mathcal{V}_1^{\psi}|$. 
% For every pair of variables $X,Y$ in $\mathcal{V}_1^{\psi}$ such that
% $MX \neq MY$, let $u_{XY}$ be any element in the symmetric difference
% of $MX$ and $MY$, and put
% \[
% D_{0} \defAs \{u_{XY}: X,Y \in \mathcal{V}_1^{\psi}
% \hbox{ and } MX \neq MY\}\,.
% \]
% In general we have $|D_{0}| = \Theta(|\mathcal{V}_1^{\psi}|^{2})$.
% However, by applying the procedure \emph{Distinguish} described in
% \cite{CanFer1995}, we can force $|D_{0}| \leq |\mathcal{V}_1^{\psi}|
% -1$, which is what we do.

%\footnote{Plainly, $|\Delta| = {\calf t
%O}(|\mathcal{V}_1|^2)$. Nevertheless it can be shown that one can find
%a set $\Delta '$ of size ${\cal O}(|\mathcal{V}_1|)$ and such that if
%$MX \neq MY$ then $MX \cap \Delta ' \neq MY \cap \Delta '$; see
%\cite{unpow}.}.

Then, after initializing $D^*$ with the set $\{Mx : x \in
\mathcal{V}_0^{\psi} \} \cup (D_{0}\cup D_1)$, for each $\varphi \in \Phi$ of
the form $(\forall z_1) \ldots (\forall z_n) \varphi_0$ having
$Z_1,\ldots,Z_m$ as arguments and for each ordered $m$-tuple
$(X_{i_1},\ldots,X_{i_m})$ of variables in $\mathcal{V}_1^{\psi}$ such
that $\model \not\models
\varphi_{X_{i_1},\ldots,X_{i_m}}^{Z_1\;\,,\ldots,\;Z_m}$, we insert in
$D^*$ elements $u_1,\ldots,u_n \in D$ such that
$\model[z_1/u_1,\ldots,z_n/u_n] \not\models
({\varphi_0})_{X_{i_1},\ldots,X_{i_m}}^{Z_1\;,\ldots,\:Z_m}$.

% 

%\centerline{------------------------}

% 
From the previous construction it follows easily that
\begin{equation}\label{DStar}
|D^*| \leq |\mathcal{V}_0^{\psi}| + (l_{\psi}+2)|\mathcal{V}_1^{\psi}| -1 +
N \cdot |\mathcal{V}_1^{\psi}|^M \cdot |\Phi|\, ,
\end{equation}
where $M$ and $N$ are, respectauto.ively, the maximal number of quantifiers
in purely universal formulae of level $1$ occurring in $|\Phi|$ and the maximal number of
quantifiers in purely universal formulae of level $0$ occurring in purely universal formulae of level $1$ in $|\Phi|$.  Thus, in general, the domain of the small
model $D^*$ is exponential in the size of the input formula $\psi$.

\subsubsection{Correctness of the relativization.}
Let us put $\modelStar = \Rel(\model, D^*, d^*, \mathcal{V}_0^{\psi}, \mathcal{V}_1^{\psi}, l_{\psi})$. We have to
show that, if $\model\models \psi$, then $\modelStar \models\psi$.

\begin{theorem}\label{correctness}
Let $\model$ be a $\TLQST$-interpretation satisfying a normalized
$\TLQSTR$-conjunction $\psi$.  Further, let $\modelStar = \Rel(\model,
D^*,d^*, \mathcal{V}_0^{\psi}, \mathcal{V}_1^{\psi}, l_{\psi})$ be the $\TLQST$-interpretation defined
according to Definition \ref{relintrp}, where $D^*$ is constructed as
above, $\mathcal{V}_0^{\psi}$ and $\mathcal{V}_1^{\psi}$ are the collections of variables of levels $0$ and $1$ occurring in $\psi$, respectively, and $l_{\psi}$ is the smallest number such that $k \leq l_{\psi}$, for every finite enumeration $\{x_1,\ldots,x_k\}$ of level $0$ variables  occurring in $\psi$.  Then $\modelStar \models \psi$.
\end{theorem}
\begin{proof}
We have to prove that $\modelStar \models \psi'$, for every conjunct
$\psi'$ in $\psi$.  Each conjunct $\psi'$ is of one of the three types
(\ref{was1}), (\ref{was2}), and (\ref{was3}) introduced in Section
\ref{normal3LQS}.  By applying Lemmas \ref{wasLe_basic}, \ref{quantifiedformOne}, or
\ref{quantifiedformTwo} to every $\psi'$ in $\psi$ (according to the type
of $\psi'$) we obtain the thesis.

Notice that the hypotheses of Lemmas \ref{wasLe_basic}, \ref{quantifiedformOne}, and
\ref{quantifiedformTwo} are fulfilled by the construction of $D^*$
outlined above. Indeed,
\begin{enumerate}
   \item $Z_{1},\ldots,Z_{m} \notin \mathcal{V}_1^{\psi}$;
   
   \item
   $X \in \mathcal{V}_1^{\psi}$, for every variable $X \in
   \mathcal{V}_1$ occurring free in $\psi$;
   
   \item
   $Mx \in D^*$, for every $x \in \mathcal{V}_{0}$ occurring free in
   $\psi$;
   
   \item
   $M^* X = MX$, if $|MX| \leq l_{\psi}$ and $|M^*X|>l_{\psi}$ otherwise, for every $X \in  \mathcal{V}_1^{\psi}$;
      
   \item
   $(MX \mathop{\Delta} MY) \cap D^{*} \neq \emptyset$, for all $X,Y
   \in \mathcal{V}_1^{\psi}$ such that $MX \neq MY$;
   
   \item
   each occurrence of finite enumeration $\{x_1,\ldots,x_k\}$ in $\psi$, with $x_i \in \mathcal{V}_{0}$, for every $i \in \{1,\ldots,k\}$, is such that $k \leq l_{\psi}$; 
      
   \item
   for every purely universal formula $(\forall z_1) \ldots (\forall
   z_n) \varphi_0$ of level $0$ occurring in a purely universal formula of level $1$, and
   variables $X_{1}, \ldots,X_{m} \in \mathcal{V}_{1}^{\psi}$ such that
   $\model \not\models ((\forall z_1) \ldots (\forall z_n)
   \varphi_0)_{X_1,\ldots,X_m}^{Z_1,\ldots,Z_m}$, there are
   $u_1,\ldots,u_n \in D^*$ such that
   $\model[z_1/u_1,\ldots,z_n/u_n]\not\models
   (\varphi_0)_{X_1,\ldots,X_m}^{Z_1,\ldots,Z_m}$;
 \item
 for every purely universal formula $(\forall z_1) \ldots (\forall
   z_n) \varphi_0$ of level $0$ occurring in $\varphi_{1}$, $\{z_1,\ldots,z_n\} \in \mathcal{V}_{0} \setminus \mathcal{V}_{0}^{\psi}$.
\end{enumerate}
\end{proof}

From the above reduction and relativization steps, it is not hard to
derive the following result:
\begin{corollary}
    The fragment $\TLQSTR$ enjoys a small model property (and
    therefore its satisfiability problem is solvable). \qed
\end{corollary}
Reasoning as in \cite{CanNic13}, it is possible to define a class of subtheories $(\TLQSTR)^h$ of $\TLQSTR$, whose formulae have quantifier prefixes of length bounded by the constant $h \geq 2$ and satisfy certain syntactic constraints, having an \textsf{NP}-complete satisfiability problem. Such subtheories are quite expressive, in fact several set-theoretic constructs treated in Section \ref{expressiveness} such as, for instance, some variants of the powerset operator can be represented in them. Moreover, it can be shown that the modal logic \textsf{S5} can be represented in $(\TLQSTR)^3$.

\section{Expressiveness of the language $\TLQSTR$}\label{expressiveness}
Several constructs of elementary set theory are easily expressible
within the language $\TLQSTR$.  In particular, it is possible to
express with $\TLQSTR$-formulae a restricted variant of the set former,
which in turn allows one to express other significant set operators
such as binary union, intersection, set difference, set
complementation, 
%the singleton and finite enumeration operators, 
the powerset operator and some of its variants, etc.

More specifically, a set former of the form $X=\{z : \varphi(z)\}$ can
be expressed in $\TLQSTR$ by the formula
\begin{equation}
\label{eq_setformX} (\forall z)(z \in X \leftrightarrow \varphi(z))\,,
\end{equation}
(in which case it is called an \emph{admissible set former of level 
$0$ for} \TLQSTR) provided that after transforming it into prenex normal
form one obtains a formula satisfying the syntactic constraints of
$\TLQSTR$.  This, in particular, is always the case whenever
$\varphi(z)$ is a quantifier-free formula of $\TLQSTR$.

In \ref{tableSetFormers1} some examples of formulae
expressible by admissible set formers of level $0$ for \TLQSTR are reported, where
$\mathbf{0}$ and $\mathbf{1}$ stand respectively for the empty set and
for the domain of the discourse, and $\overline{\phantom{Y}}$ is the
complementation operator with respect to the domain of the discourse.
The formulae in the first column of \ref{tableSetFormers1}
%, but
%the last one, 
are the allowed atoms in the fragment \TLS (Two-Level
Syllogistic) which has been proved decidable in \cite{FerOm1978}. 
Since $\{x_1,\ldots,x_k\} = X$ is a level $0$ quantifier-free atomic formula in \TLQSTR,  \TLS with finite enumerations turns out to be expressible by 
\TLQSTR-formulae.

\begin{table}[t]{\small
\begin{center}
    \begin{tabular}{c|c}
        & \emph{admissible set formers for \TLQSTR of level $0$} \\
    \hline \\[-.2cm]
    $X = \mathbf{0}$ & $X = \{z : z \neq z\}$\\
    $X = \mathbf{1}$ & $X = \{z : z = z\}$\\
    $X = \overline{Y}$ & $X = \{z : z \notin Y\}$\\
    $X = Y_{1} \cup Y_{2}$ & $X = \{z : z \in Y_{1} \vee z \in
    Y_{2}\}$ \\
    $X = Y_{1} \cap Y_{2}$ & $X = \{z : z \in Y_{1} \wedge z \in
    Y_{2}\}$ \\
    $X = Y_{1} \setminus Y_{2}$ & $X = \{z : z \in Y_{1} \wedge z \notin
    Y_{2}\}$ \\
% 
%    $X = \{x_{1},\ldots,x_{k}\}$ & $X = \{z : z = x_{1} \vee \ldots 
%    \vee z = x_{k}\}$\\[.2cm]
% 
    \hline
\end{tabular}
\end{center}
\caption{\label{tableSetFormers1}Some literals expressible by
admissible set formers of level $0$ for \TLQSR.}}
\end{table}

In addition to the formulae in \ref{tableSetFormers1} the 
following literals 
\begin{equation}
\label{furtherLiterals}
 Z \subseteq X\,, \quad |Z| \leq h\,, \quad |Z| < h + 1\,, \quad
|Z| \geq h + 1\,, \quad |Z| = h
\end{equation}
are also expressible by \TLQSTR-formulae of level $0$, where $|\cdot|$
denotes the cardinality operator and $h$ stands for a nonnegative
integer constant (cf.~\ref{tableSetFormers2}).  In fact, it
turns out that all literals (\ref{furtherLiterals}) can be expressed
by level $0$ purely universal \TLQSTR-formulae which are linked to the
variable $Z$, so that they can freely be used in the matrix
$\varphi(Z)$ of a level $1$ universal formula of the form $(\forall
Z)\varphi(Z)$.  Let us consider, for instance, the formula
\begin{equation}
    \label{ZKey} 
(\forall z_{1})\ldots(\forall z_{h+1})
\Bigg({\displaystyle\bigwedge_{1 \leq i \leq h+1} z_i \in Z
\rightarrow \bigvee_{1 \leq i < j \leq h+1} z_i = z_j }\Bigg)
\end{equation}
which expresses the literal $|Z| \leq h$. The linkedness condition 
for it relative to the variable $Z$ is
\[
\neg\Bigg(\bigwedge_{1 \leq i \leq h+1} z_i \in Z
\rightarrow \bigvee_{1 \leq i < j \leq h+1} z_i = z_j\Bigg)
\rightarrow \bigwedge_{1 \leq i \leq h+1} z_i \in Z\,,
\]
which is plainly a valid \TLQSTR-formula since it is an instance of the
propositional tautology $\neg (\mathbf{p} \rightarrow \mathbf{q})
\rightarrow \mathbf{p}$, showing that (\ref{ZKey}) is linked to the
variable $Z$.  Similarly, one can show that the remaining formulae in
(\ref{furtherLiterals}) can also be expressed by level $0$ purely
universal \TLQSTR-formulae which are linked to the variable $Z$.

\begin{table}[t]{\small
\begin{center}
    \begin{tabular}{c|c}
        & $\TLQSTR$-formulae \\
    \hline \\[-.2cm]
    $Z \subseteq X$ & $(\forall z) (z \in Z \rightarrow z \in X)$\\
    $|Z| \leq h$ & $(\forall z_{1})\ldots(\forall z_{h+1})
    \Bigg({\displaystyle\bigwedge_{1 \leq i \leq h+1} z_i \in Z
    \rightarrow \bigvee_{1 \leq i < j \leq h+1} z_i = z_j
    }\Bigg)$\\
    $|Z| < h + 1$ & $|Z| \leq h$\\
    $|Z| \geq h + 1$ & $\neg (|Z| < h + 1) $\\
    $|Z| \geq 0$ & $Z = Z$\\
    $|Z| = h$ & $|Z| \leq h \wedge |Z| \geq h $\\[.2cm]
    \hline
\end{tabular}
\end{center}
\caption{\label{tableSetFormers2}Further formulae expressible by
\TLQSTR-formulae of level $0$.}}
\end{table}

Similar remarks apply also to the set former of the form $A=\{Z :
\varphi(Z)\}$.  This can be expressed by the $\TLQSTR$-formula 
\begin{equation}
\label{eq_setformA}  (\forall Z)(Z \in A \leftrightarrow \varphi(Z))
\end{equation}
(in which case it is called an \emph{admissible set former of level
$1$ for} \TLQSTR) provided that
%\begin{itemize}
%\item[(a)] 
$\varphi(Z)$ does not contain any quantifier over variables
of sort $1$, and 
%\item[(b)] 
all quantified variables of sort $0$ in $\varphi(Z)$ are
linked to the variable $Z$ as specified in
condition~(\ref{condition}).
%\end{itemize}

\begin{table}[t]{\small
\begin{center}
    \begin{tabular}{c|c}
        & \emph{admissible set formers of level $1$ for \TLQSTR} \\
    \hline \\[-.2cm]
    $A = \mathbf{0}$ & $X = \{Z : Z \neq Z\}$\\
    $A = \mathbf{1}$ & $X = \{Z : Z = Z\}$\\
    $A = \overline{B}$ & $A = \{Z : Z \notin B\}$\\
    $A = B_{1} \cup B_{2}$ & $A = \{Z : Z \in B_{1} \vee Z \in
    B_{2}\}$\\
    $A = B_{1} \cap B_{2}$ & $A = \{Z : Z \in B_{1} \wedge Z \in
    B_{2}\}$\\
    $A = B_{1} \setminus B_{2}$ & $A = \{Z : Z \in B_{1} \wedge Z \notin
    B_{2}\}$\\
    $A = \{X_{1},\ldots,X_{k}\}$ & $A = \{Z : Z = X_{1} \vee \ldots 
    \vee Z = X_{k}\}$\\
% % 
%     $A \subseteq B$ & $B = A \cup B$\\
% 
    $A = \pow(X)$ & $A = \{Z : Z \subseteq X\}$\\
    $A = \pow_{\leq h}(X)$ & $A = \{Z : Z \subseteq X \wedge |Z| \leq 
    h\}$\\
    $A = \pow_{= h}(X)$ & $A = \{Z : Z \subseteq X \wedge |Z| = 
    h\}$\\
    $A = \pow_{\geq h}(X)$ & $A = \{Z : Z \subseteq X \wedge |Z| \geq 
    h\}$\\
    $A = \pow_{< h+1}(X)$ & $A = \{Z : Z \subseteq X \wedge |Z| \leq 
    h\}$\\%[.2cm]
    $\cdots$ & $\cdots$\\
%     $\cdots$ & $\cdots$\\
    \hline
\end{tabular}
\end{center}
\caption{\label{tableSetFormers3}Some literals expressible by
admissible set formers of level $1$ for \TLQSTR.}}
\end{table}

Some examples of formulae expressible by admissible set formers of
level $1$ for \TLQSTR are reported in \ref{tableSetFormers3}.  In
this case the symbol $\mathbf{1}$ stands for the powerset of the
domain of the discourse.  The meaning of the overloaded symbol
$\mathbf{1}$ can always be correctly disambiguated from the context.
In view of the fact that, as already remarked, the literals
(\ref{furtherLiterals}) can be expressed by level $0$ purely universal
\TLQSTR-formulae which are linked to the variable $Z$, it follows that 
all set formers in  \ref{tableSetFormers3} are indeed 
admissible.

The propositional combination of the following literals 
\begin{equation}\label{3LSSP}
\begin{array}{llllll}
A = \mathbf{0}\,, &  A = \mathbf{1}\,, & A = \overline{B} \,,
& A = B_{1} \cup B_{2}\,,\\  
A = B_{1} \cap B_{2}\,, & A = B_{1} \setminus B_{2}\,, & A =
\{X_{1},\ldots,X_{k}\}\,, & A = \pow(X)
\end{array}
\end{equation}
present in the first column of \ref{tableSetFormers3} form the
proper fragment \TLSSP of the theory \TLSSPU (Three-Level Syllogistic
with Singleton, Powerset, and Unionset) whose decision problem has
been solved in \cite{CanCut93}.  We recall that in addition to the
formulae in (\ref{3LSSP}), \TLSSPU involves also unionset clauses of 
the form $X = \bigcup A$, with $X$ a variable of sort $1$ and $A$ a 
variable of sort $2$, which, however, are not expressible by 
\TLQSTR-formulae. 

Besides the ordinary powerset operator, \TLQSTR-formulae allow one also
to express the variants $\pow_{\leq h}(X)$, $\pow_{= h}(X)$, and
$\pow_{\geq h}(X)$ reported in \ref{tableSetFormers3}, which
denote respectively the collection of all the subsets of $X$ with at
most $h$ distinct elements, with exactly $h$ elements, and with at
least $h$ distinct elements.  It is interesting to observe that the
satisfiability problem for the propositional combination of literals
of the forms $x \in y$, $x = y \cup z$, $x = y \cap z$, $x =
y \setminus z$,
involving also one occurrence of literals of the form $y = \pow_{=
1}(x)$, has been proved to be decidable in \cite{CanFerMic87}, when
sets are interpreted in the standard von Neumann hierarchy (cf.\ 
\cite{jech2003set}).

Another interesting variant of the powerset operator is the $\pow^*$
operator introduced in \cite{Can91,CanOmoUrs02} in the solution to the
satisfiability problem for the extension of \MLS with the powerset and
singleton operators.  We recall that given sets $X_{1}, \ldots, 
X_{k}$, $\pow^*(X_1,\ldots,X_k)$ denotes
the collection of all subsets of $\bigcup_{i=1}^k X_i$ which have
nonempty intersection with each set $X_{i}$, for $i = 1,\ldots,k$. In 
symbols,
\[
\begin{array}{rcl}
    \pow^*(X_1,\ldots,X_k) &\defAs& \displaystyle\left\{Z : Z
    \subseteq \bigcup_{i=1}^k X_i \wedge \bigwedge_{i=1}^{k} Z \cap
    X_i \neq \emptyset \right\} \\
    & \!\!\!\!\!= & \displaystyle\left\{Z : Z
    \subseteq \bigcup_{i=1}^k X_i \wedge \bigwedge_{i=1}^{k} \neg (Z 
    \subseteq \overline X_i) \right\}\,.
\end{array}
\]
From the latter expression, it readily follows that the literal $A =
\pow^*(X_1,\ldots,X_k)$ can be expressed by a \TLQSTR-formula.

%The unordered Cartesian product literal $A = X_1 \otimes \ldots
%\otimes X_n$, 
%
%
%\subsubsection{Other set theoretical constructs expressible in $\TLQSR$}\label{applsetcomplex} 

%\subsubsection{The unordered Cartesian product.}\label{applsetcomplex}

Given sets $X_{1},\ldots,X_{n}$, the unordered Cartesian product $X_1
\otimes \ldots \otimes X_n$ is the set
\[X_1 \otimes \ldots \otimes X_n \defAs \Big\{\{x_{1},\ldots,x_{n}\} :
x_{1} \in X_{1}, \ldots, x_{n} \in X_{n}\Big\}\,.\]
Then, the literal
\begin{equation}
    \label{uCp}
    A = X_1 \otimes \ldots \otimes X_n\,,
\end{equation}
where $A$ stands for a variable of level 2 and $X_{1},\ldots,X_{n}$
here stand for variables of level 1, can be expressed by the $\TLQSTR$-formula
\begin{equation}
    \label{not3LQSR}
    (\forall Z)\Bigg(Z \in A \longleftrightarrow (\exists x_{1})
    \ldots (\exists x_{n})\Bigg(\bigwedge_{i=1}^{n} x_{i} \in X_{i}
    \wedge \{x_{1},\ldots,x_{n}\} = Z\Bigg)\Bigg)\,.
\end{equation}
%However, (\ref{not3LQSR}) is not a \TLQSR-formula, since unordered
%$n$-tuples are not admitted in the \TLQSR language.
%%%%%%%%%%%%%%%%%%%%%%%%%%%%%%%%%%%%%%%%%%%%%%%%%%%%%
In what follows, we show that (\ref{not3LQSR})  can be expressed without making use of the finite enumeration operator. 
When the sets $X_{1},\ldots,X_{n}$ are pairwise disjoint or, on the
opposite side, when they all coincide, we can readily express the
literal (\ref{uCp}) by a \TLQSR-formula.
For instance, if the sets $X_{1},\ldots,X_{n}$ are pairwise
disjoint, then $Z \in X_1 \otimes \ldots \otimes X_n$ if and only
if
\begin{enumerate}
    \item\label{first}  $|Z| = n$, and
    \item\label{second}  there exist $x_{1} \in X_{1},\ldots,x_{n} \in X_{n}$ such
    that $x_{1} \in Z$, \ldots, $x_{n} \in Z$\,.
%     $S = \{x_{1}, \ldots, x_{n}\}$\,.
\end{enumerate}
The above conditions can be used to express the literal (\ref{uCp}) by
the following \TLQSR-formula

$
(\forall Z)\Bigg(Z \in A \longleftrightarrow \Bigg(|Z| = n \wedge (\exists
x_{1}) \ldots (\exists x_{n})\Bigg(\bigwedge_{i=1}^{n} (x_{i} \in X_{i}
\wedge x_{i} \in Z)\Bigg)\Bigg)\Bigg)\,,
$

as is easy to check, where

$
\begin{array}{rcl}
    |Z| = n & \equivAs & |Z| \leq n \wedge |Z| \geq n  \\
    |Z| \leq n & \equivAs & (\forall x_{1})\ldots(\forall
x_{n+1})\Bigg(\displaystyle\bigwedge_{i=1}^{n+1} x_{i} \in Z \rightarrow \bigvee_{1
\leq i < j \leq n+1} x_{i} = x_{j}\Bigg)  \\
    |Z| \geq n & \equivAs & \neg(|Z| \leq n-1)
\end{array}
$
% \begin{multline*}
% |Z| = n \equivAs (\forall x_{1})\ldots(\forall
% x_{n+1})\Bigg(\bigwedge_{i=1}^{n+1} x_{i} \in Z \rightarrow \bigvee_{1
% \leq i < j \leq n+1} x_{i} = x_{j}\Bigg)\\
% {} \wedge (\exists x_{1})\ldots(\exists x_{n})\Bigg(\bigwedge_{1 \leq i < j
% \leq n} x_{i} \neq x_{j} \wedge \bigwedge_{i=1}^{n} x_{i} \in
% Z\Bigg)\,.
% \end{multline*}
(notice that $|Z| \leq n$ is linked to the variable $Z$).

%\medskip

When $X_{1} = \ldots = X_{n}$, then $Z \in X_1 \otimes \ldots \otimes
X_n$ if and only if
$
|Z| \leq n  \text{ and } Z \subseteq X_{1}.
$
Thus, in this particular case, the literal (\ref{uCp}) can be
expressed by the \TLQSR-formula

$
(\forall Z)\Big(Z \in A \longleftrightarrow \Big(|Z| \leq n \wedge
(\forall x)(x \in Z \rightarrow x \in X_{1}) \Big)\Big)\,.
$
% where
% \[
% |Z| \leq n \equivAs (\forall x_{1})\ldots(\forall
% x_{n+1})\Bigg(\bigwedge_{i=1}^{n+1} x_{i} \in Z \rightarrow
% \bigvee_{1 \leq i < j \leq n+1} x_{i} = x_{j}\Bigg)\,.
% \]

However, if we make no assumption on the sets
$X_{1},\ldots,X_{n}$, in order to characterize the sets $Z$ belonging
to $X_1 \otimes \ldots \otimes X_n$ by a \TLQSR-formula, we have to
consider separately the cases in which $|Z| = n$, $|Z| = n-1$, etc.,
listing explicitly, for each of them, all the allowed membership
configurations of the members of $Z$.  For instance, if
$n=2$, we have $Z \in X_1 \otimes X_{2}$ if and only if
\begin{itemize}
    \item $|Z| = 2$ and there exist distinct $x_{1} \in
    X_{1}$ and $x_{2} \in X_{2}$ s. t.
    $x_{1}, x_{2} \in Z$; or

    \item $|Z| = 1$ and the intersection $X_{1} \cap X_{2} \cap Z$ is
    nonempty.
\end{itemize}
Thus the following \TLQSR-formula expresses the literal $A = X_1
\otimes X_2$:

$
\begin{array}{ll}
(\forall Z)\Big(Z \in A \longleftrightarrow \Big( \Big(\hspace{-4pt}&
|Z| = 2 \wedge (\exists x_{1})(\exists x_{2})\Big(x_{1} \neq x_{2}
\wedge \displaystyle\bigwedge_{i=1}^{2}(x_{i} \in X_{i} \wedge x_{i}
\in Z) \Big)\Big)\\
& {} \vee \Big(|Z| = 1  \wedge (\exists
x_{1})( x_{1} \in X_{1} \wedge
x_{1} \in X_{2} \wedge x_{1} \in Z)\Big)
\Big)\Big)\,.
\end{array}
$

Likewise, in the case $n=3$, we have $Z \in X_1
\otimes X_{2} \otimes X_3$ if and only if
\begin{itemize}
    \item $|Z| = 3$ and there exist pairwise distinct $x_{1} \in
    X_{1}$, $x_{2} \in X_{2}$, and $x_{3} \in X_{3}$ such that
    $x_{1}, x_{2} , x_{3} \in Z$; or

    \item $|Z| = 2$ and there exist distinct $x_{1}$ and $x_{2}$ such
    that  either
    \begin{itemize}
	\item $x_{1} \in X_{1} \cap X_{2}$ and $x_{2} \in X_{3}$, or
	
	\item $x_{1} \in X_{1} \cap X_{3}$ and $x_{2} \in X_{2}$, or
	
	\item $x_{1} \in X_{2} \cap X_{3}$ and $x_{2} \in X_{1}$,
    \end{itemize}
    and such that $x_{1}, x_{2} \in Z$; or
    \item  $|Z| = 1$ and the intersection $X_{1} \cap X_{2}
    \cap X_{3} \cap Z$ is nonempty.
\end{itemize}

%% We refrain from writing explicitly the corresponding \TLQSR-formula,
%% since it is quite lengthy.  Rather, more in general, we show how to
%% write the \TLQSR-formulae expressing the literals $A = X_1 \otimes
%% \ldots \otimes X_n$, for $n=1,2,\ldots$, we make use of the following
%% elementary property.
%%
%% write in a compact manner the \TLQSR-formula expressing the literal
%% $A = X_1 \otimes X_2 \otimes X_3$ and, more in general, the
%% \TLQSR-formulae expressing the literals $A = X_1 \otimes \ldots
%% \otimes X_n$, for $n=1,2,\ldots$, we make use of the following
%% elementary property.
%
%More in general, we have the following lemma, proved in \cite{CanNicArxiv15}.

\begin{lemma}
    \label{characterization}
Let $X_{1},\ldots,X_{n}$ be given sets.
Then $Z \in X_1 \otimes \ldots \otimes X_n$ if and only there
exists a partition $P$ of the set $\{1,\ldots,n\}$ and a
bijection $\sigma: Z \rightarrow P$ such that
\begin{equation}
    \label{condition1}
    \text{if $i \in \sigma(x)$, then $x \in X_{i}$, for $x \in Z$
    and $i \in \{1,\ldots,n\}$}.
 \end{equation}\qed
\end{lemma}
\begin{proof}
Let $Z \in X_1 \otimes \ldots \otimes X_n$. Then $Z =
\{x_{1},\ldots,x_{n}\}$, for some $x_{1}\in X_{1},\ldots,x_{n}\in
X_{n}$. For $x \in Z$, let us put
\[
\sigma(x) \defAs \{i : x = x_{i}\}\,.
\]
Then it is an easy matter to check that $P \defAs \{\sigma(x) : x \in
X\}$ is a partition of $\{1,\ldots,n\}$ and $\sigma$ is a bijection
from $Z$ into $P$ which satisfies (\ref{condition1}).

Conversely, assume that $\sigma: Z \rightarrow P$ is a bijection
satisfying (\ref{condition1}), for a partition $P$ of $\{1,\ldots,n\}$
and a set $Z$, and put
\[
x_{i} \defAs \sigma^{-1}(P_{i})\,,
\]
where $P_{i}$ is the block of $P$ containing $i$. Then it plainly
follows that $x_{i} \in X_{i}$, for $i=1,\ldots,n$ and that
$Z=\{x_{1},\ldots,x_{n}\}$, proving that $Z \in X_1 \otimes \ldots
\otimes X_n$.
\end{proof}
%
%%\bigskip
%
Let $\mathfrak{P}_{n}$ be the collection of all partitions of the set
$\{1,\ldots,n\}$. For any partition $P \in \mathfrak{P}_{n}$, we will
assume that the blocks $b_{1}(P),\ldots,b_{|P|}(P)$ of $P$ are
ordered by a total order $\prec$ in such a way that

$
b_{i}(P) \prec b_{j}(P) \quad \text{ if and only if } \quad \min
b_{i}(P) < \min b_{j}(P)\,.
$

Then, based on Lemma~\ref{characterization}, the literal $A = X_{1}
\otimes \ldots \otimes X_{n}$ is expressed by the following
\TLQSR-formula
\begin{multline}
    \label{mostro}
(\forall Z)\Bigg(Z \in A \leftrightarrow \bigwedge_{P \in \mathfrak{P}_{n}}
\Bigg(|Z|=|P| \wedge (\exists z_{1})\ldots(\exists
z_{|P|})\Bigg(\bigwedge_{1\leq i < j \leq |P|} z_{i} \neq z_{j} \\
 \wedge
\bigwedge_{i=1}^{|P|} \Bigg(z_{i} \in Z \wedge \bigwedge_{j \in b_{i}(P)}
z_{i} \in X_{j}\Bigg)\!\!\Bigg)\!\!\Bigg)\!\!\Bigg)\,.
\end{multline}
Let $\ell_{n}$ be the length of the formula (\ref{mostro}).
Then the following bounds on $\ell_{n}$ hold:
\begin{equation}\label{bounds1}
\ell_{n} = \Omega(nB_{n})\,, \qquad \ell_{n} =
\mathcal{O}(n^{2}B_{n})\,,
\end{equation}
where $B_{n} = |\mathfrak{P}_{n}|$ is the $n$th Bell's number.  Using
the bounds on $B_{n}$ by Berend and Tassa (cf.\ \cite{BT10})

$
\left(\frac{n}{e\ln n}\right)^{\!\!n}< B_{n} < \left(\frac{0.792n}{\ln
(n+1)}\right)^{\!\!n},
$

the bounds (\ref{bounds1}) yield

$
\ell_{n} = \Omega\left(n\left(\frac{n}{e\ln n}\right)^{\!\!n}\right)\,, \qquad
\ell_{n} = \mathcal{O}\left(n^{2}\left(\frac{0.792n}{\ln
(n+1)}\right)^{\!\!n}\right).
$

\section{Conclusions and future work}\label{conclusions}
We have presented a three-sorted stratified set-theoretic fragment,
$\TLQSTR$, and given a decision procedure for its satisfiability
problem.  The fragment turns out to be quite expressive since it allows to represent
several set-theoretic construct such as variants of the powerset operator and 
the unordered Cartesian product. Thanks to the presence of the finite enumeration operator, 
$\TLQSTR$ allows to represent the unordered Cartesian product by means of a formula which is linear in the
size of the product. Another representation of the latter construct is possible without resorting to 
the finite enumeration operator, but is this case the formula turns out to be exponentially longer. 

Proceeding as in \cite{CanNic13} it is possible to single out a family $\{(\TLQSTR)^h\}_{h \geq 2}$ of
sublanguages of $\TLQSTR$, characterized by imposing further
constraints in the construction of the formulae,  such that each
language in the family has the satisfiability problem
\textsf{NP}-complete, and to show that the modal logic $\Sc$ can be
formalized in $(\TLQSTR)^{3}$.
%
%Techniques to translate modal formulae in set theoretic terms have already
%been proposed in \cite{BaMo96}, in the context of hyperset theory, and in \cite{DMP95}
%in the ambit of weak set theories not involving the axiom of extensionality and the axiom
%of foundation.
%
We further intend to study the possibility of formulating non-classical logics
in the context of well-founded set theory constructing suitable extensions
of the $\TLQSTR$ fragment. 

%In particular, we would like 
%to introduce in our language a notion of ordered
%pair and the operation of composition for binary relations.
%
%%The algorithm provided here serves only for the purpose of
%%proving the decidability of $\TLQSR$ and is not by any means an
%%efficient procedure.
%%
%%We are currently working on the construction of a more effe
%
%
%
%
%
%
%%tableau-based decision procedure for $\TLQSR$ in the flavour of
%%\cite{Can97,BeckertHartmer98}.
%
We also plan to extend the language so as
it can express the set theoretical construct of general union, thus being able to
subsume the theory \TLSSPU. 
% Another direction of future investigations concerns
%$n$-sorted languages involving also
%constructs to express ordered $n$-uples of individuals.
%
% %in order to obtain a more efficient
%%decision procedure that could be embedded in automated theorem
%%provers and proof verifiers.
%
%%A further direction of investigation is the relationship between
%%$3LQS$ and decidable fragments of monadic second-order logic, the
%%latter being employed in fields like automata theory
%%\cite{Buechi62}, complexity theory \cite{Fagin74}, but also in more
%%applicative fields like verification of hardware
%%\cite{BasinKlarlund95}, software and distributed systems
%%\cite{TACAS-GJJ+}.

\bibliographystyle{plain}

\end{document}